\documentclass{article}
\usepackage[utf8]{inputenc}
\usepackage{amsmath}
\usepackage{amssymb}
\usepackage{amsthm}
\usepackage{authblk}
\usepackage{appendix}
\usepackage{comment}
\usepackage{subfigure}

\title{Geometry of the set of synchronous quantum correlations}
\author{Travis B. Russell}
\affil{Army Cyber Institute, \\ United States Military Academy, \\ West Point, NY}
\affil{\textit{travis.russell@westpoint.edu}}
\date{}

\usepackage[numbers]{natbib}
\usepackage{graphicx}

\newtheorem{theorem}{Theorem}[section]
\newtheorem{lemma}[theorem]{Lemma}
\newtheorem{proposition}[theorem]{Proposition}
\newtheorem{corollary}[theorem]{Corollary}
\newtheorem{definition}[theorem]{Definition}

\newtheorem{remark}[theorem]{Remark}
\newtheorem{question}[theorem]{Question}

\newcommand*\tr{\text{tr}}
\newcommand*\Tr{\text{Tr}}
\newcommand*\co{\text{co}}
\newcommand*\Ran{\text{Ran}}

\begin{document}

\maketitle

\begin{abstract}
We provide a complete geometric description of the set of synchronous quantum correlations for the three-experiment two-outcome scenario. We show that these correlations form a closed set. Moreover, every correlation in this set can be realized using projection valued measures on a Hilbert space of dimension no more than 16.
\end{abstract}

\section{Introduction}

One of the fundamental challenges of quantum mechanics is that a quantum state cannot be directly observed. To obtain information about an unknown quantum state, we can perform measurements and record the results. The outcome of any such measurement is statistically determined by the quantum state. Thus by performing many measurements one can begin to understand some aspects of the behavior of the quantum state by examining the resulting probability distribution. When the state is entangled and the measurements are performed on separate subsystems, we obtain a joint probability distribution known in the literature as a quantum correlation.

It is a well-known and fundamental result that the set of quantum correlations which can be achieved with an entangled state is strictly larger than the set of quantum correlations which can be achieved by a separable state \cite{bell_epr} (separable states are elementary tensors of the form $\phi \otimes \psi$ in a tensor product of Hilbert spaces $H_a \otimes H_b$, while entangled states are states that cannot be expressed in this form). This observation has led to many interesting developments in quantum information theory, some of which have potentially intriguing applications in the fields of quantum communication and quantum cryptography (for example, see Bennett-Brassard \cite{bb84}).

While the distinction between correlations generated by separable states and those generated by entangled states is well-established, a complete understanding of the latter set is still lacking, even in the three-experiment two-outcome setting. Much of the research regarding the geometry of these correlations focuses on the winning probabilities of certain non-local games, particularly the $I_{3322}$ game (see \cite{VidickI3322}, \cite{Collins_2004}, \cite{Froissart1981}, and \cite{osti_21448443}). It is not known if the maximum winning probability for this game can be achieved over the set of three-experiment two-outcome quantum correlations. If it could be shown that the maximal value cannot be achieved, then it would follow that the quantum correlation sets are not closed in this setting. Dykema-Paulsen-Prakash \cite{DeltaGame} have shown that a synchronous version of the $I_{3322}$ game does achieve its maximal winning probability over the synchronous part of the three-experiment two-outcome quantum correlations, raising the possibility that this set could be topologically closed.

In this paper we aim to make a small contribution towards these problems by providing an explicit geometric description of the set of synchronous quantum correlations in the case of three experiments with two outcomes each. We determine that this set is topologically closed - a conclusion which is perhaps surprising in light of several recent proofs of the non-closure of the quantum correlation sets in general (see Slofstra \cite{Slofstra1} and Dykema-Paulsen-Prakash \cite{DPP1}). Moreover, we demonstrate that every quantum correlation in this setting can be achieved with projections on a Hilbert space of dimension no more than 16. All results are obtained using only tools from linear algebra and Euclidean geometry, though we appeal to some well-known results about quantum correlation sets along the way. Our approach is largely inspired by the geometric approach used in Dykema-Paulsen-Prakash \cite{DPP1}.

Another motivation for explicitly computing quantum correlation sets comes from operator algebras. The combined results of Junge et. al. \cite{JMPPSW2011}, Fritz \cite{FritzKirchberg}, and Ozawa \cite{MR3067294} showed that Connes' embedding conjecture, a long-standing problem in operator algebras, is equivalent to the conjecture that the closure of the set of quantum correlations is equal to the set of so-called quantum commuting correlations for all possible numbers of experiments and outcomes. It was recently shown that the same is true if one considers only synchronous correlation sets \cite{MR3432742}, \cite{MR3776034}. It follows that one could, in principle, settle Connes' conjecture by providing a complete description of both the quantum correlation sets and the quantum commuting correlation sets for all possible numbers of experiments and outcomes. If Connes' conjecture is false, then one need only compute some quantum correlation set and demonstrate a quantum commuting correlation which does not lie in the closure of this set. While we draw no conclusion about Connes' conjecture in this paper, the computation of the synchronous quantum correlations for the three-experiment two-outcome scenario provides new data that could be examined to find or rule out counterexamples to Connes' conjecture in this setting.

We should mention a few papers from the literature related to the question of computing the geometry of the quantum correlation sets. The original proof that the set of quantum correlations is not closed is due to Slofstra \cite{Slofstra1}. Another proof \cite{DPP1} shows that the set of synchronous quantum correlations is not closed when the number of experiments exceeds five and the number of outcomes is at least two. In addition, the authors provide an explicit description of a continuous region in Euclidean space where the quantum correlations constitute a countable dense subset (see Remark 4.3 of Dykema-Paulsen-Prakash \cite{DPP1}). Much of the intuition behind our approach is inspired by techniques in their paper. We should also mention Goh et. al. \cite{PhysRevA.97.022104} which provides a fairly detailed description of the quantum correlations in the two-experiment two-outcome case. Finally, a preprint of Thinh, Varvitsiotis and Cai  \cite{ThinhStructure} provides an explicit description of a related set, the quantum correlators, for a family of experiment-outcome scenarios. We note that the quantum correlators, as defined by these authors, differs from the correlation sets we are concerned with, as explained in section II of their paper \cite{ThinhStructure}.

Our paper is organized as follows. In Section 2, we summarize relevant concepts and results from the literature on quantum correlation sets. We also define the basic tools we will be using and apply them to the two-experiment two-outcome scenario as an example. In Section 3, we derive a description of the synchronous quantum correlation sets for three experiments and two outcomes. That section is divided into subsections focusing on different types of quantum correlations, each subsection building on the results of the previous subsection. Finally, in Section 4, we provide a few concluding remarks concerning Connes' embedding conjecture and non-local games.

We conclude this introduction with a summary of the mathematical notation used. We let $\mathbb{C}^d$, $\mathbb{R}^d$, and $\mathbb{M}_d$ denote the $d$-dimensional complex Hilbert space, the $d$-dimensional real Hilbert space, and the set of all $d \times d$ complex matrices, respectively. Throughout, we will identify operators on the Hilbert space $\mathbb{C}^d$ with $d \times d$ matrices in the obvious way, working over the canonical basis of $\mathbb{C}^d$ unless another basis is specified. Given matrices $A$ and $B$, we let $A \oplus B$ denote the direct sum, i.e. \[ A \oplus B := \begin{bmatrix} A & 0 \\ 0 & B \end{bmatrix}, \] and we let $A \otimes B$ denote the Kronecker product of $A$ and $B$. We also let $\vec{0}_d$ denote the vector of zeros in $\mathbb{C}^d$ or $\mathbb{R}^d$, and we let $0_d$ denote the $d \times d$ zero matrix and $I_d$ denote the $d \times d$ identity matrix. Given a $n \times m$ matrix $A$, we let $A^\dagger$ denote the $m \times n$ conjugate transpose of $A$. A square matrix $P$ is called a projection if $P=P^\dagger$ and $P^2 = P$. By a projection valued measure, we mean a set $\{P_i\}_{i=1}^m$ of projections with the property that $\sum_{i=1}^m P_i = I_d$. We use $\Tr(\cdot)$ and $\tr_d(\cdot)$ for the ordinary matrix trace and the normalized matrix trace (i.e., $\tr_d(\cdot) = \tfrac{1}{d}\Tr(\cdot)$), respectively. For sets $S,T \subseteq \mathbb{R}^d$, we let $\co\{S,T\}$ denote the (not necessarily closed) convex hull of $S$ and $T$ in $\mathbb{R}^d$. Finally, given an integer $i$, we let $\delta_{k,i}$ denote the Kronecker delta function (i.e. $\delta_{k,i}=1$ if $i = k$, and $\delta_{k,i}=0$ otherwise).

\section{Preliminaries}
Suppose two parties, Alice and Bob, are performing probabilistic experiments. For our purposes, we will assume each of Alice and Bob can perform one of $n$ experiments and that each experiment has $m$ possible outcomes. We will let the quantity $p(i,j|x,y)$ represent the probability that Alice obtains outcome $i$ and Bob obtains outcome $j$ given that Alice performed experiment $x$ and Bob performed experiment $y$. We call the tensor $(p(i,j|x,y))_{i,j,x,y}$ a \textbf{correlation} if it satisfies \[ \sum_{i,j} p(i,j|x,y) = 1\] for every choice of $x$ and $y$. Let us further assume that Alice and Bob are spatially separated and unable to pass signals to each other. This is modeled mathematically by adding the restriction that the marginal densities \[ p_a(i|x) := \sum_j p(i,j|x,y), \quad p_b(j|y) := \sum_i p(i,j|x,y) \] are well defined - that is, the matrix $p_a$ is independent of the choice of $y$ and $p_b$ is independent of the choice of $x$. Such a correlation is called \textbf{non-signaling} and the set of all non-signaling correlations is denoted by $C_{ns}(n,m)$.

We may further restrict Alice and Bob's capabilities by assuming that their correlations arise from a combination of deterministic strategy and shared randomness. Specifically, let $\{\lambda(1), \dots, \lambda(k)\}$ be a discrete probability distribution, and assume that for each $x \leq n$ and $t \leq k$, Alice possesses a deterministic distribution $p_a(i|x,t)$ (i.e. $p_a(i|x,t) \in \{0,1\}$), and similarly Bob possesses deterministic distributions $p_b(j|y,t)$. Then the formula \[ p(i,j|x,y) := \sum_t \lambda(t) p_a(i|x,t) p_b(j|y,t) \] defines a non-signaling correlation. We call any correlation of this form a \textbf{local} correlation, and we denote the set of all local correlations by $C_{loc}(n,m)$.

Our primary interest is in correlations which lie between the local and non-signaling correlations, namely the \textbf{quantum} correlations. Assume that Alice has access to a finite-dimensional Hilbert space $H_a$ and Bob has access to a finite-dimensional Hilbert space $H_b$. Let $\phi \in H_a \otimes H_b$ be a unit vector. Let us further assume that Alice and Bob share the possibly entangled state $\phi$, and are able to perform measurements on their respective Hilbert spaces. Specifically, for each $x$ we assume Alice possesses a projection valued measure $\{E_{x,i}\}_{i=1}^m$ and likewise Bob possesses projection valued measures $\{F_{y,j}\}_{j=1}^m$ for each $y \leq n$. Then the correlation defined by \[ p(i,j|x,y) = \langle \phi, E_{x,i} \otimes F_{y,j} \phi \rangle \] is a non-signaling correlation. Any correlation defined in this way is called a quantum correlation, and we let $C_q(n,m)$ denote the set of all quantum correlations.

In general, the correlation sets described above are convex and satisfy \[ C_{loc}(n,m) \subseteq C_q(n,m) \subseteq C_{ns}(n,m) \subseteq \mathbb{R}^{n^2 m^2}. \] All inclusions in the above sequence are known to be strict. It is of historical importance that $C_q(n,m) \neq C_{loc}(n,m)$ in general. In fact, $C_q(2,2) \neq C_{loc}(2,2)$ as a consequence of the CHSH inequality \cite{CHSH}. The local correlations describe the behavior of particles in a universe governed by the theory of local hidden variables espoused by Einstein, Podolski and Rosen \cite{EPR}, whereas the set of quantum correlations describe the behavior of particles in a universe governed by Von Neumann's formalism of quantum mechanics. John Bell first showed that these sets are distinct \cite{bell_epr}, and the experimental verification of this fact has been hailed as evidence that particles obey the laws of quantum mechanics \cite{BellTheoremExperiment}.


In this paper, we are primarily interested in the set of \textbf{synchronous} correlations. A correlation is synchronous if, for all $x$, $p(i,j|x,x) = 0$ whenever $i \neq j$. For each $r \in \{loc, q, ns\}$ we let $C_r^s(n,m)$ denote the set of synchronous correlations. It is clear that $C_r^s(n,m) \subseteq C_r(n,m)$, since $C_r^s(n,m)$ is obtained by intersecting $C_r(n,m)$ with the hyperplane in $\mathbb{R}^{n^2m^2}$ defined by the synchronous relations $p(i,j|x,x)=0$ for $i \neq j$.

The following result will be employed freely throughout. Recall that a \textbf{$C^*$-algebra} is a closed self-adjoint algebra of bounded operators on a Hilbert space. A \textbf{tracial state} $\tau$ on a $C^*$-algebra $\mathfrak{A}$ is a linear functional $\tau: \mathfrak{A} \rightarrow \mathbb{C}$ mapping positive operators to positive real numbers and satisfying $\tau(1)=1$ and $\tau(xy)=\tau(yx)$ for all $x,y \in \mathfrak{A}$.

\begin{theorem}[Theorem 5.5 / Corollary 5.6, Paulsen e.t al. \cite{MR3460238}] \label{PaulsenWinter} A correlation $p \in C_q(n,m)$ is synchronous if and only if there exists a finite-dimensional $C^*$-algebra $\frak{A}$ and projection valued measures $\{E_{x,i}\}_{i=1}^m \subset \frak{A}$ and a tracial state $\tau$ on $\frak{A}$ such that \[ p(i,j|x,y) = \tau(E_{x,i} E_{y,j}). \] \end{theorem}


It is a consequence of the Artin-Wedderburn theorem that every finite-dimensional $C^*$-algebra is isomorphic to a direct sum of matrix algebras (for example, see Theorem III.1.1 of Davidson's textbook\cite{MR1402012}). Moreover, each matrix algebra possesses a unique tracial state $\tr_d: \mathbb{M}_d \rightarrow \mathbb{C}$ defined by $\tr_d(x) = \frac{1}{d} \Tr(x)$, where $\Tr(\cdot)$ is the usual matrix trace. Consequently, whenever $\tau$ is a trace on $\frak{A} \cong \mathbb{M}_{d_1} \oplus \dots \oplus \mathbb{M}_{d_k}$, we may assume that $\tau = \sum_{i=1}^k \lambda_i \tr_{d_i}$ where $\sum \lambda_i = 1$ - i.e. $\tau$ is a convex combination of normalized matrix traces. Furthermore, whenever $\{E_{x,i}\}_{i=1}^m \subset \mathbb{M}_d$ are projection valued measures, we have \[ \tr_d (E_{x,i} E_{y,j}) =  \langle \phi_d, E_{x,i} \otimes E_{y,j}^T \phi_d \rangle \] where $\phi_d = \sum_{k=1}^d \tfrac{1}{\sqrt{d}} e_k \otimes e_k$ is a maximally entangled state in $\mathbb{C}^d \otimes \mathbb{C}^d$ (where $\{e_k\}$ is the canonical basis for $\mathbb{C}^d$). Let $C_{max}^s(n,m)$ denote the set of quantum correlations defined by $p(i,j|x,y) = \tr_d(E_{x,i} E_{y,j})$ for $\{E_{x,i}\} \subset \mathbb{M}_d$ and for some $d$ (or equivalently $p(i,j|x,y) = \langle \phi_d, E_{x,i} \otimes E^T_{y,j} \phi_d \rangle$ for the maximally entangled state $\phi_d$). Then we have the following (see Theorem 9 of Lackey-Rodriguez\cite{Lackey}, Corollary 5.5 of Lupini et. al.\cite{PerfectStrategies}, and Theorem 3.7 of Alhajjar-Russell\cite{MaxEntangle}).

\begin{theorem} \label{densityThm} Let $p \in C_q^s(n,m)$. Then there exist $t_1,\dots,t_k$ with each $t_i \geq 0$ and $\sum t_i = 1$ and  correlations $p_1,\dots,p_k \in C_{max}^s(n,m)$ such that $p = \sum_{i=1}^k t_i p_i$. Hence, \[ C_q^s(n,m) = \co\{C_{max}^s(n,m)\}. \] Moreover, we have \[ \overline{C_q^s(n,m)} = \overline{C_{max}^s(n,m)}. \] \end{theorem}

We will further restrict our attention to subsets of the quantum correlations with fixed marginal density matrices $p_a(i|x), p_b(j|y)$. To specify such a subset, we need only specify the values of $p_a(i|x)$ or $p_b(j|y)$. Indeed, whenever $p \in C_q^s(n,m)$, we have \[ p_a(i|x) = \sum_j p(i,j|x,x) = \sum_i p(i,j|x,x) = p_b(i|x). \] Hence, we may dispense with the subscripts and consider only the marginal matrix $p(i|x) \in \mathbb{R}^{nm}$.

Let $\vec{r} \in \mathbb{R}^{nm}$ with entries indexed as $r_{i,x}$, $i \leq m$, $x \leq n$. Then we define the $\vec{r}$-slice of $C_r^s(n,m)$ ($r \in \{loc, q, max, ns\}$) by \[ S_{\vec{r}}[ C_r^s(n,m) ] := \{ p \in C_r^s(n,m) : p(i|x) = r_{i,x} \}. \] Clearly $S_{\vec{r}}[ C_r^s(n,m) ]$ is non-empty if and only if $\sum_{i=1}^m r_{i,x} = 1$ for each $x$, and $r_{i,x} \geq 0$ for each $i$ and $x$. Furthermore, observe that $S_{\vec{r}}[ C_{max}^s(n,m) ]$ is non-empty if and only if $\vec{r} \in \mathbb{Q}^{nm}$ in addition to $\sum_{i=1}^m r_{i,x} = 1$ for each $x$ (where $\mathbb{Q}$ denotes the rational numbers). This is because for $p \in C_{max}^s(n,m)$, \[ p(i|x) = \tr_d (E_{x,i}) \in \mathbb{Q} \] since the trace of a projection is its rank. It is evident from these definitions that for each $t \in \{ loc, q, ns \}$ the $\vec{r}$-slices of $C_t^s(n,m)$ are all convex. It is not clear whether or not the rational $\vec{r}$-slices of $C_{max}^s(n,m)$ are convex, though they are closed under rational convex combinations.

We will be especially interested in determining the structure of the set $S_{\vec{r}}[ C_{max}^s(n,m)]$. Our interest is due to the following observations. By Theorem \ref{densityThm}, $C_q^s(n,m) = \co\{C_{max}^s(n,m)\}$. Furthermore, \[ C_{max}^s(n,m) = \bigcup_{\vec{r}} S_{\vec{r}}[ C_{max}^s(n,m) ]. \] In other words, $C_{max}^s(n,m)$ is a countable disjoint union of its slices. It follows that $C_q^s(n,m) = \co\{ \cup_{\vec{r}} S_{\vec{r}}[ C_{max}^s(n,m) ]\}$. So by describing the geometry of each slice $S_{\vec{r}}[C_{max}^s(n,m)]$ we can determine the geometry of $C_q^s(n,m)$.

Henceforth we will focus on the case when $m = 2$, where the possible outcomes are $\{0,1\}$. In this case there are several simplifying assumptions that can be made. First, the marginal density matrix $p(i|x)$ can be reduced to the vector $\vec{r}=(p(0|1),p(0|2),\dots,p(0|n))$. This is because $p(1|x) = 1 - p(0|x)$, so we only need to know the value of $p(0|x)$ for each $x$ to determine the marginal density matrix. Furthermore, for each fixed $x,y \leq n$ with $x < y$, the matrix $p(i,j|x,y)$ has the form \[ \begin{pmatrix} w_{x,y} & r_x - w_{x,y} \\ r_y - w_{x,y} & w_{x,y} + (1 - r_x - r_y) \end{pmatrix} \] where $w_{x,y} = p(0,0|x,y)$, $r_x = p(0|x)$ and $r_y = p(0|y)$ (this is a consequence of the non-signaling conditions). Hence, the entire matrix is determined by the values $r_x, r_y$ and $w_{x,y}$. For $y < x$, we have $p(i,j|x,y) = p(j,i|y,x)$, using Theorem \ref{PaulsenWinter} and the observation that \begin{eqnarray} p(i,j|a,b) & = & \tau(E_{a,i}E_{b,j}) \nonumber \\ & = & \tau(E_{b,j}E_{a,i}) \nonumber \\ & = & p(j,i|b,a) \nonumber \end{eqnarray} for any tracial state $\tau$. It follows that $C_q^s(n,2)$ is entirely determined by the values $r_x = p(0,x)$ and $w_{x,y} = p(0,0|x,y)$ for $x < y$. Thus the dimension of $C_q^s(n,2)$ is at most $\tfrac{n(n+1)}{2}$, and the dimension of each slice $S_{\vec{r}}[C_q^s(n,2)]$ is at most $\tfrac{n(n-1)}{2}$. Consequently, to understand the geometry of $S_{\vec{r}}[C_q^s(n,2)]$ it suffices to consider the projection $\hat{S}_{\vec{r}}[C_q^s(n,2)]$ whose components are given by the upper triangular entries of the matrix $(p(0,0|x,y))_{x,y \leq n}$, a subset of $\mathbb{R}^{\tfrac{n(n-1)}{2}}$.

To determine the geometry of the slice $\hat{S}_{\vec{r}}[C_q^s(n,2)]$, we will first consider the geometry of the slice $\hat{S}_{\vec{r}}[C_{max}^s(n,2)]$. To compute this, we will consider the subset of $\hat{S}_{\vec{r}}[C_q^s(n,2)]$ generated by projections on Hilbert spaces of fixed dimension $d$ which we denote by $S_d(N_1, N_2, \dots, N_n)$, where $N_i$ is the rank of the $i$-th projection. We summarize these definitions in the following.

\begin{definition} Let $n \geq 2$. Then for each $\vec{r} \in [0,1]^n$ we define \[ \hat{S}_{\vec{r}}[C_q^s(n,2)] = \{ (p(0,0|x,y))_{x < y \leq n} : \quad p(0|i)=r_i, \quad p \in C_q^s(n,2) \} \] where $r_i$ are the entries of $\vec{r}$. Moreover, for integers $N_1, N_2, \dots, N_n \leq d$, we define \[ S_d(N_1, N_2, \dots, N_n) = \{ p \in \hat{S}_{\vec{r}}[C_q^s(n,2)] : \quad p(x,y) = \tr_d(P_x P_y), \quad \tr_d(P_i) = \tfrac{N_i}{d}, x < y \leq n \} \] where $P_1, P_2, \dots, P_n$ are $d \times d$ projections. \end{definition}

As a short illustration, we use the above ideas to quickly compute the set $C_q^s(2,2)$. Indeed, the geometry of $C_q(2,2)$ is well understood\cite{PhysRevA.97.022104}, although we are not aware of an explicit formulation in the synchronous case. To perform the computation, we need one lemma which we will use later in the paper as well. The result is probably well-known, but we provide a proof for completeness.

\begin{lemma} \label{traceProjections} Let $n_1$, $n_2$ and $d$ be integers with $n_1,n_2 \leq d$. Then \[ S_d(n_1,n_2) = [ \max(0, \tfrac{n_1+n_2 - d}{d}), \min(\tfrac{n_1}{d},\tfrac{n_2}{d})]. \] \end{lemma}

\begin{proof} It suffices to consider the case when $n_1 \leq n_2$. Let $P$ and $Q$ be projections on $\mathbb{C}^d$ of rank $n_1$ and $n_2$, respectively. Since $Tr(\cdot)$ is a positive linear map, we see that the functional $\phi_P(x) = Tr(Px)$ is also a positive linear map since $x \geq 0$ implies that $PxP \geq 0$ and $Tr(Px) = Tr(PxP)$. It follows that $\Tr(PQ) \geq 0$. Since $Q \leq I_d$, we also have $\Tr(PQ) \leq \Tr(P I_d) = n_1$. In the case when $n_1 + n_2 > d$, observe that if the range of $P$ is $V$ and the range of $Q$ is $W$, then \begin{eqnarray} \dim(V \cap W) & = & \dim(V) + \dim(W) - \dim(V+W) \nonumber \\ & \geq & n_1 + n_2 - d. \nonumber \end{eqnarray} If $R$ is the projection onto $V \cap W$, then $R \leq P$ and $R \leq Q$, implying that $Tr(PQ) \geq Tr(RQ) \geq Tr(R^2) = Tr(R) \geq n_1 + n_2 - d$.

Now let $P_1 = I_{n_1} \oplus 0_{d-n_1}$ and $Q_1 = I_{n_2} \oplus 0_{d-n_2}$. Then $\Tr(P_1Q_1) = n_1$. If $n_1 + n_2 \leq d$, then taking $Q_2 = 0_{d-n_2} \oplus I_{n_2}$ we get $\Tr(P_1 Q_2) = 0$. By holding $P_1$ fixed and replacing $Q_2$ with $UQ_2U^\dagger$, where $U$ varies over the group of unitary operators on $\mathbb{C}^d$, we obtain all other values in the interval $[0, n_1]$. Finally, suppose $n_1 + n_2 > d$. Then $\Tr(P_1Q_2) = n_1 + n_2 - d$, and by again replacing $Q_2$ with $UQ_2U^\dagger$ we get all values in $[ n_1 + n_2 - d, n_1]$. \end{proof}

To describe $C_q^s(2,2)$, it suffices to describe each slice $\hat{S}_{\vec{r}}[C_{max}^s(2,2)]$ and then compute the convex hull. By Lemma \ref{traceProjections}, we see that \[ \hat{S}_{\vec{r}}[C_{max}^s(2,2)]=[ \max(0, r_1 + r_2 - 1), \min(r_1,r_2)] \] whenever $\vec{r}=(r_1,r_2)$ with $r_1,r_2 \in \mathbb{Q}$. By Theorem \ref{densityThm}, this implies that \[ \hat{S}_{\vec{r}}[\overline{C_q^s(2,2)}] = [\max(0,r_1+r_2-1),\min(r_1,r_2)] \] for every $\vec{r}=(r_1,r_2) \in [0,1]^2$. However, we can show that $$[\max(0,r_1+r_2-1),\min(r_1,r_2)] \subset \hat{S}_{\vec{r}}[C_q^s(2,2)].$$ For example, if $r_1 \leq r_2 \leq 1$ and $r_1+r_2 - 1 \geq 0$, we can generate \[ [\max(0,r_1+r_2-1),\min(r_1,r_2)] \] as follows. Set $P_1 = 1 \oplus 0 \oplus 0, Q_1 = 1 \oplus 1 \oplus 0 \in \mathbb{C} \oplus \mathbb{C} \oplus \mathbb{C}$ and define a trace $\tau_1$ on $\mathbb{C} \oplus \mathbb{C} \oplus \mathbb{C}$ via $\tau(a \oplus b \oplus c) = r_1 a + (r_2 - r_1)b + (1-r_2)c$. Then $\tau_1(P_1 Q_1) = r_1$. Next set $P_2 = 1 \oplus 1 \oplus 0, Q_2 = 1 \oplus 0 \oplus 1 \in \mathbb{C} \oplus \mathbb{C} \oplus \mathbb{C}$, and define a trace $\tau_2$ on $\mathbb{C} \oplus \mathbb{C} \oplus \mathbb{C}$ via $\tau_2(a \oplus b \oplus c) = (r_1 + r_2 - 1)a + (1 - r_2)b + (1-r_1)c$. Then $\tau_2(P_2 Q_2) = r_1 + r_2 - 1$. Similar arguments for the various types of $\vec{r}$-slices show that every slice is in $C_q^s(2,2)$. Thus we conclude that $C_q^s(2,2)$ is closed and is an affine image of the three dimensional body \[ \{ (r_1, r_2, [\max(0,r_1+r_2-1),\min(r_1,r_2)]) \}_{r_1,r_2 \in [0,1]}. \]

\section{The three-experiment two-outcome case}
In this section we aim to provide a complete description of the set $C_q^s(3,2)$. Our strategy will be to mimic the argument used to describe $C_q^s(2,2)$ in the previous section. We first make a few preliminary observations. Assume that $p \in C_q^s(3,2)$. Then the entries of $p$ are completely determined by the six values $r_1=p(0|1), r_2=p(0|2), r_3=p(0|3), w_{1,2}=p(0,0|1,2), w_{1,3}=p(0,0|1,3)$ and $w_{2,3}=p(0,0|2,3)$, as explained in the previous section. Throughout this section, we will order vectors in $\hat{S}_{\vec{r}}[C_q^s(3,2)]$ as $(w_{1,2}, w_{1,3}, w_{2,3})$, and similarly for vectors in $S_d(n_1,n_2,n_3)$. To simplify notation, we will also write $S_d(n) = S_d(n,n,n)$.

We call a vector $\vec{r}=(r_1,r_2,r_3)$ \textbf{standard} if $0 \leq r_1 \leq r_2 \leq r_3 \leq 1/2$. When $\vec{r}$ is standard, we call the corresponding slice $S_{\vec{r}}[C_q^s(3,2)]$ a \textbf{standard slice}. It is evident that every $\vec{r}$-slice can be obtained from a standard slice by some combination of ``reversing outcomes" and ``swapping experiments". More specifically, if for $x,y \in \{1,2,3\}$, $x \neq y$, we define $\epsilon_{x,y}: C_q^s(3,2) \rightarrow C_q^s(3,2)$ to be the map that interchanges experiments $x$ and $y$ (for example, $\epsilon_{x,y}(p)(i,j|x,y) = p(i,j|y,x)$), and if for each $x \in \{1,2,3\}$ we set $\pi_x$ to be the map that reverses the outcomes of experiment $x$ (so that, for example, $\pi_x(p)(1,0|y,x) = (1,1|y,x)$), then an arbitrary slice of $C_q^s(3,2)$ is easily seen to be the image of a standard slice under some composition of $\epsilon_{x,y}$'s and $\pi_x$'s. It is also evident that the $\epsilon_{x,y}$ and $\pi_x$ maps are invertible affine maps - hence, every slice of $C_q^s(3,2)$ is an affine image of a standard slice of $C_q^s(3,2)$.

We will further subdivide the standard slices into three types. We call a standard vector $\vec{r}=(r_1,r_2,r_3)$ \textbf{type I} if $r_1=r_2=r_3$, \textbf{type II} if $r_1 \leq r_2=r_3$ and \textbf{type III} if $r_1 \leq r_2 \leq r_3$. Likewise, we call the corresponding $\vec{r}$-slices type I, II, or III if $\vec{r}$ is type I, II, or III, respectively. Our analysis of the slices of $C_q^s(3,2)$ will proceed as follows. We will first determine the structure of the type I slices of $C_{max}^s(3,2)$. We will then determine the structure of the type II slices of $C_{max}^s(3,2)$ by describing their structure in terms of the type I slices. Finally we will determine the structure of the type III slices of $C_{max}^s(3,2)$ by describing their structure in terms of the type II slices. Having determined the structure of all the slices of $C_{max}^s(3,2)$, we will use the fact that $\overline{C_q^s(3,2)} = \overline{\co \{ C_{max}^s(3,2) \}}$ to provide a complete description of $C_q^s(3,2)$.


The next lemma will be crucial to our results. In short, it will allow us to reduce the dimension of the Hilbert space needed to implement a given correlation.


\begin{lemma} \label{mainLemma} Let $n_1,n_2,n_3,d$ be positive integers with each $n_i \leq d$. Then if $n_1 + n_2 < d$ then \[ S_d(n_1,n_2,n_3) \subseteq \tfrac{d-1}{d} \co(S_{d-1}(n_1,n_2,n_3), S_{d-1}(n_1, n_2, n_3 - 1)). \] Likewise, if $n_1 + n_3 < d$, then \[ S_d(n_1,n_2,n_3) \subseteq \tfrac{d-1}{d} \co(S_{d-1}(n_1,n_2,n_3), S_{d-1}(n_1, n_2-1, n_3)) \] and if $n_2 + n_3 < d$, \[ S_d(n_1,n_2,n_3) \subseteq \tfrac{d-1}{d} \co(S_{d-1}(n_1,n_2,n_3), S_{d-1}(n_1 - 1, n_2, n_3)). \] \end{lemma}

\begin{proof} We need only prove the first statement, since the others clearly follow by symmetry. So assume $n_1 + n_2 < d$. Recall that for any pair of subspaces $W, V \subseteq \mathbb{C}^d$, we have $\dim(V \cap W) = \dim(V) + \dim(W) - \dim(V + W)$. Consequently, \begin{eqnarray} \dim(\ker(P_1) \cap \ker(P_2)) & \geq & (d - n_1) + (d - n_2) - d \nonumber \\
& = & d - (n_1 + n_2) > 0. \nonumber \end{eqnarray} So there exists some unit vector $u \in \ker(P_1) \cap \ker(P_2)$. By expanding $\{u\}$ to an orthonormal basis for $\mathbb{C}^d$, we may assume that the projections $P_i$ have the following form as matrices: \[
P_{i} = \begin{bmatrix}
\tilde{P}_{i} & \vec{0}_{d-1} \\
\vec{0}_{d-1}^\dagger & 0
\end{bmatrix}
\text{ for } i \in \{1,2\}, \quad
P_3 = \begin{bmatrix}
B & \vec{v} \\
\vec{v}^\dagger & t
\end{bmatrix} \] where $\tilde{P}_i,B \in \mathbb{M}_{d-1}$, $\vec{v} \in \mathbb{C}^{d-1}$, and $t \in [0,1]$. Since $P_1$ and $P_2$ are projections, $\tilde{P}_{1}$ and $\tilde{P}_2$ are $(d-1) \times (d-1)$ projections of rank $n_1$ and $n_2$, respectively.

Since $P_3$ is a projection, we may use the relation $P_3^2 = P_3$ to further decompose $P_3$ as \begin{equation} \label{lemmaDecomp}
P_3 =  \begin{bmatrix}
\tilde{P}_3 & \vec{0} \\
\vec{0}^\dagger & 0
\end{bmatrix} +
\begin{bmatrix}
(1-t)\vec{w}\vec{w}^\dagger & (t-t^2)^{1/2} \vec{w} \\
(t-t^2)^{1/2} \vec{w}^\dagger & t
\end{bmatrix} \end{equation} for some unit vector $\vec{w} \in \mathbb{C}^{d-1}$ and some $d \times d$ rank $n_3 - 1$ projection $\tilde{P}_3$ orthogonal to $\vec{w} \vec{w}^\dagger$. Indeed, the equation $P_3^2 = P_3$ implies that \begin{equation} \label{lemmaMatrix} \begin{bmatrix} B & \vec{v} \\ \vec{v}^\dagger & t \end{bmatrix} = \begin{bmatrix} B^2 + \vec{v} \vec{v}^\dagger & B \vec{v} + t \vec{v} \\ \vec{v}^\dagger B + t \vec{v}^\dagger & \|\vec{v}\|^2 + t^2 \end{bmatrix}. \end{equation} From the upper right corner of (\ref{lemmaMatrix}), we get $B \vec{v} = (1-t) \vec{v}$. From the lower right of (\ref{lemmaMatrix}), we see that $\| \vec{v} \|^2 = t - t^2$. Set $\vec{w} = (t-t^2)^{-1/2} \vec{v}$, a unit vector. Since $B$ is positive semidefinite, $B = \tilde{P}_3 + (1-t) \vec{w} \vec{w}^\dagger$, where $\tilde{P}_3$ is positive semidefinite with $\tilde{P}_3 \vec{w} = \vec{0}$. Finally, the upper left of (\ref{lemmaMatrix}) and $B = \tilde{P}_3 + (1-t) \vec{w} \vec{w}^\dagger$ yields \begin{eqnarray} \tilde{P}_3 + (1-t) \vec{w} \vec{w}^\dagger & = & (\tilde{P}_3 + (1-t) \vec{w} \vec{w}^\dagger)^2 + (t-t^2)\vec{w} \vec{w}^\dagger \nonumber \\ & = & \tilde{P}_3^2 + (1-t)^2 \vec{w}\vec{w}^\dagger + (t-t^2) \vec{w}\vec{w}^\dagger \nonumber \\ & = & \tilde{P}_3^2 + (1-t) \vec{w} \vec{w}^\dagger. \nonumber \end{eqnarray} We conclude that $\tilde{P}_3 = \tilde{P}_3^2$, so $\tilde{P}_3$ is a projection orthogonal to $\vec{w}\vec{w}^\dagger$.

Using the decomposition (\ref{lemmaDecomp}), we have for $k = 1,2$ \begin{eqnarray} \Tr(P_k P_3) & = & \Tr(\tilde{P}_k \tilde{P}_3) + (1-t) \Tr(\tilde{P}_k \vec{w}\vec{w}^\dagger) \nonumber \\ & = & t \Tr(\tilde{P}_k \tilde{P}_3) + (1-t) \Tr(\tilde{P}_k (\tilde{P}_3 + \vec{w}\vec{w}^\dagger)). \nonumber \end{eqnarray} It follows that \begin{eqnarray} (\tr_d(P_1 P_2), \tr_d(P_1 P_3), \tr_d(P_2 P_3)) & = & t \tfrac{d-1}{d} (\tr_{d-1}(\tilde{P}_1 \tilde{P}_2), \tr_{d-1}(\tilde{P}_1 \tilde{P}_3), \tr_{d-1}(\tilde{P}_2 \tilde{P}_3)) \nonumber \\ & + & (1-t) \tfrac{d-1}{d} (\tr_{d-1}(\tilde{P}_1 \tilde{P}_2), \tr_{d-1}(\tilde{P}_1 \hat{P}_3), \tr_{d-1}(\tilde{P}_2 \hat{P}_3)) \nonumber \end{eqnarray} where $\hat{P}_3 := \tilde{P}_3 + \vec{w} \vec{w}^\dagger$. Since $\tilde{P}_3$ is orthogonal to the rank one projection $\vec{w}\vec{w}^\dagger$, $\hat{P}_3$ is a rank $n_3$ projection, so the statement follows. \end{proof}

We will sometimes need to increase the dimension of the Hilbert space used to implement a correlation. Though the following is easy to prove, we record it here since we will use this fact frequently.

\begin{lemma} \label{LemmaBlowUp} For all positive integers $n_1,n_2,n_3,d$ with each $n_i \leq d$ we have \begin{eqnarray} S_d(n_1,n_2,n_3) & \subseteq & \tfrac{d+1}{d} S_{d+1}(n_1+1, n_2, n_3), \nonumber \\ S_d(n_1,n_2,n_3) & \subseteq & \tfrac{d+1}{d} S_{d+1}(n_1, n_2+1, n_3), \nonumber \end{eqnarray} and \begin{eqnarray} S_d(n_1,n_2,n_3) & \subseteq & \tfrac{d+1}{d} S_{d+1}(n_1, n_2, n_3+1). \nonumber \end{eqnarray} \end{lemma}

\begin{proof} By symmetry, it suffices to prove \[ S_d(n_1,n_2,n_3) \subseteq \tfrac{d+1}{d} S_{d+1}(n_1+1, n_2, n_3). \] Assume that $P_i$ are rank $n_i$ projections on $\mathbb{C}^d$. Then we can define new projections $\tilde{P}_i$ on $\mathbb{C}^{d+1}$ by setting $\tilde{P}_1 = P_1 \oplus 1$, $\tilde{P}_2 = P_2 \oplus 0$ and $\tilde{P}_3 = P_3 \oplus 0$. Then clearly \[ \tr_d(P_i P_j) = \tfrac{d+1}{d} \tr_{d+1}(\tilde{P}_i \tilde{P}_j) \] for each $i$ and $j$. \end{proof}

\subsection{Type I slices}

We begin by studying the structure of the type I slices. We start with the case of the $\vec{r}$-slice for $\vec{r}=(1/2,1/2,1/2)$. The first proof below is based on the proof of Theorem 1 in Tsirelson's paper\cite{Tsirel'son1987}. We provide the proof for completeness, and since Tsirelson's result is phrased in a different context. We only consider the three-experiment case here, although the idea generalizes to any number of experiments. Before stating the proposition, we recall that the $n \times n$ elliptope is defined to be the set of $n \times n$ positive semidefinite matrices over $\mathbb{R}$ with diagonal entries equal to 1.

\begin{proposition} \label{PropHalfSlice} Let $\vec{r} = (1/2,1/2,1/2)$. Then \[ \hat{S}_{\vec{r}}[C_q^s(3,2)] = S_{2n}(n) = S_2(1) \] for every $n \geq 1$. Moreover, $S_2(1)$ is an affine image of the above diagonal portion of the $3 \times 3$ elliptope. In particular, $(x,y,z) \in S_2(1)$ if and only if $(x,y,z) = \frac{1}{4}((x',y',z') + (1,1,1))$ where $x',y',z' \in \mathbb{R}$ are the above diagonal entries of a $3 \times 3$ positive semidefinite matrix with diagonal entries of 1. \end{proposition}

\begin{proof} In the $3 \times 3$ case, a matrix \[
q = \begin{pmatrix}
1 & x' & y' \\
x' & 1 & z' \\
y' & z' & 1 \end{pmatrix} \] is positive semidefinite if and only if there exist unit vectors $\vec{a}, \vec{b}, \vec{c} \in \mathbb{R}^3$ such that $x' = \langle \vec{a}, \vec{b} \rangle, y' = \langle \vec{a}, \vec{c} \rangle$, and $z' = \langle \vec{b}, \vec{c} \rangle$. Define $U_1 = a_1 X + a_2 Y + a_3 Z$, $U_2 = b_1 X + b_2 Y + b_3 Z$, and $U_3 = c_1 X + c_2 Y + c_3 Z$, where $X,Y$, and $Z$ are the standard Pauli matrices in $\mathbb{M}_2$ and $a_i, b_i$ and $c_i$ are the entries of $\vec{a}, \vec{b},$ and $\vec{c}$ (respectively) with respect to the canonical orthonormal basis of $\mathbb{R}^3$. Then each $U_i$ is a trace zero Hermitian unitary matrix. Hence, the operators $P_i := \frac{1}{2}(U_i + I)$ are rank one projections. Setting $(x,y,z)=(\tr_2(P_1 P_2), \tr_2(P_1 P_3), \tr_2(P_2 P_3))$ defines a correlation in $ \hat{S}_{\vec{r}}[ C_{max}^s(n,m) ]$. Moreover, \begin{eqnarray} \tr_2(P_i P_j) & = & \frac{1}{4}( \tr_2(U_i U_j) + 1) \nonumber \\
& = & \frac{1}{4}(q_{i,j} + 1). \nonumber \end{eqnarray} Thus, if $x', y',$ and  $z'$ are the off-diagonal entries of a matrix $q$ in the $3 \times 3$ elliptope, then $(x,y,z) = \frac{1}{4}(x'+1,y'+1,z'+1) \in S_2(1)$.

For the other direction, suppose $(x,y,z) \in \hat{S}_{\vec{r}}[ C_q^s(3,2) ]$. Then there exists a finite-dimensional $C^*$-algebra $\frak{A}$ with a tracial state $\tau$ and projections $P_1, P_2$, and $P_3$ such that $x = \tau(P_1 P_2)$, $y = \tau(P_1 P_3)$, and $z = \tau(P_2 P_3)$, with $\tau(P_i)=1/2$ for each $i=1,2,3$. Note that $\frak{A}_h$, the real vector space of hermitian elements of $\frak{A}$, forms a real inner-product space with inner product given by $\langle x, y \rangle = \tau(xy)$. Let $e$ denote the identity element in $\frak{A}_h$. By replacing $\frak{A}_h$ with the span of $\{P_1, P_2, P_3, e\}$, we obtain a real Hilbert space $H$. After selecting a basis for $H$ we may identify each $P_i$ with a vector $p_i \in \mathbb{R}^n$ for some $n$ with $\|p_i\|^2 = \tau(p_i) = 1/2$. Notice that $\langle p_i, e \rangle = \|p_i\|^2$. Let $u_i := 2p_i - e$. Then \begin{eqnarray} \|u_i\|^2 & = & \langle (2p_i - e), (2p_i - e) \rangle \nonumber \\ & = & 4\|p_i\|^2 - 2\|p_i\|^2 - 2 \|p_i\|^2 + \|e\|^2 \nonumber \\ & = & 1. \nonumber \end{eqnarray} So each $u_i$ is a unit vector. Hence, the matrix with entries $q_{i,j} = \langle u_i, u_j \rangle$ is in the $3 \times 3$ elliptope. Setting $x' = \langle u_1, u_2 \rangle$, $y' = \langle u_1, u_3 \rangle$, $z' = \langle u_2, u_3 \rangle$ and observing that $(x,y,z) = \frac{1}{4}((x',y',z') + (1,1,1))$ completes the proof. \end{proof}


Identifying each matrix in the $3 \times 3$ elliptope with the vector $(x',y',z')$ representing its off-diagonal entries, we obtain a region of $\mathbb{R}^3$ described by the equations \[ -1 \leq x',y',z' \leq 1, \] \[ 1 + 2x'y'z' - x'^2 - y'^2 - z'^2 \geq 0 \] by Sylvester's criterion. This convex region is depicted in Figure 1 below.

\begin{figure}[h!]
\centering
\includegraphics[width=.4\textwidth]{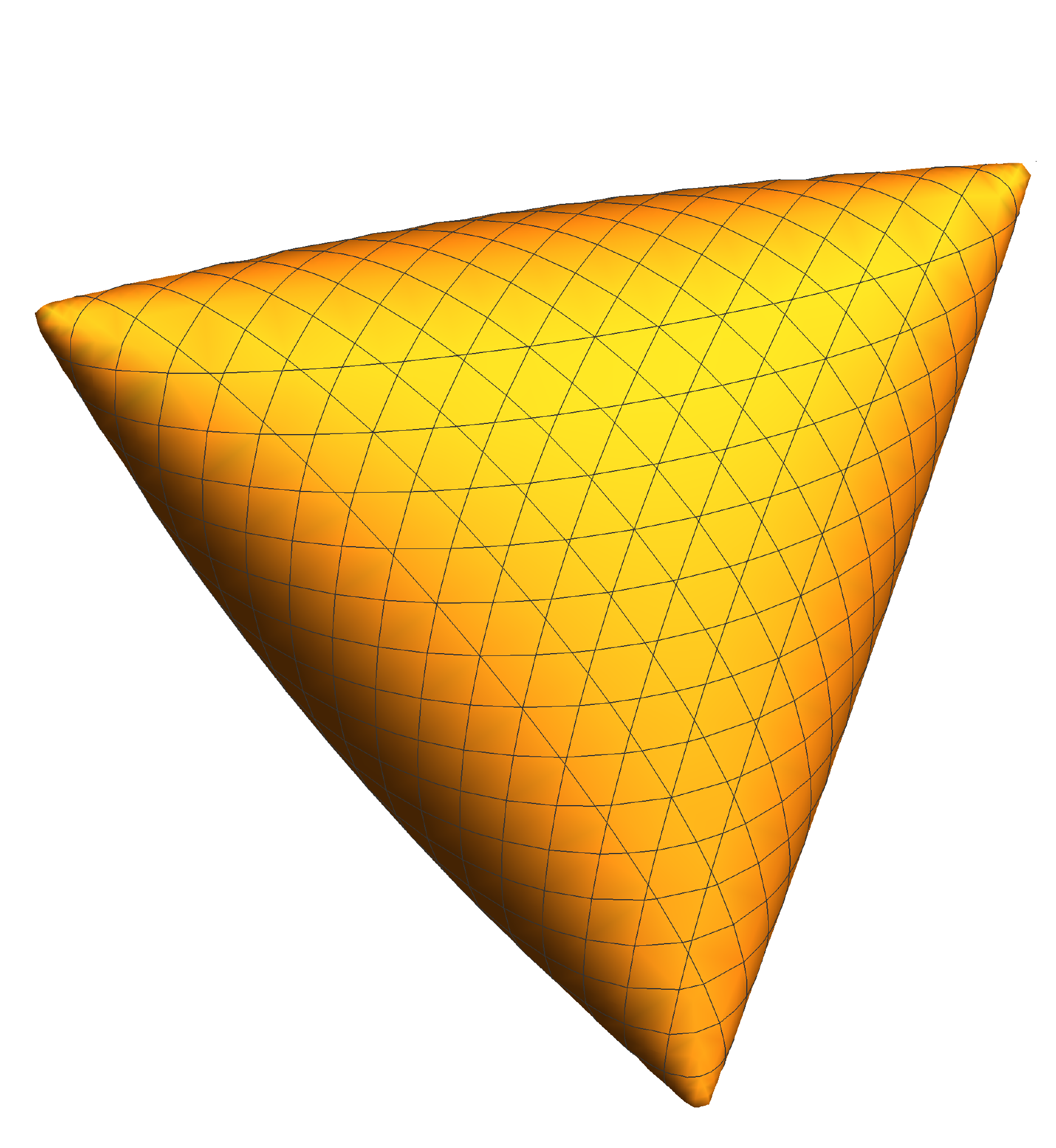}
\caption{The $(1/2,1/2,1/2)$-slice of $C_q^s(3,2)$.}
\end{figure}

We will show below in Proposition \ref{propTypeI} that type I correlations can be described in terms of the slice $S_2(1)$, which we have just shown to be an affine image of the elliptope. To this end, we must determine the structure of $S_d(n)$ for all $n \leq d/2$. We begin with an easy case.

\begin{proposition} \label{PropEasyCase} Assume that $d \geq 3$. Then \[ S_d(1) \subseteq \co\{\vec{0}, \tfrac{2}{d} S_2(1)\}. \] \end{proposition}

\begin{proof} Let $P_1,P_2,P_3$ be rank-1 projections on $\mathbb{C}^d$. Since \[ \dim( \Ran(P_1) + \Ran(P_2) + \Ran(P_3)) \leq 3 \] we may assume that each $P_i$ has a matrix of the form $P_i = \tilde{P}_i \oplus 0_{d-3}$, where $\tilde{P}_i$ is a $3 \times 3$ matrix. It follows that $S_d(1) = \frac{3}{d}S_3(1)$. Let $\hat{n_i}$ denote the row vector in $\mathbb{R}^3$ with a zero in the $n_i$ entry and ones in every other entry. Then, by Lemma \ref{mainLemma}, \[ S_3(1) \subseteq \bigcap_{i=1}^3 \frac{2}{3} \co\{S_2(1), S_2(\hat{n_i}) \} =: C. \] Now observe that \begin{eqnarray} S_2(\hat{n}_1) & = & (0,0,[0,\tfrac{1}{2}]), \nonumber \\ S_2(\hat{n}_2) & = & (0,[0,\tfrac{1}{2}],0), \nonumber \\ \text{and} \qquad S_2(\hat{n}_3) & = & ([0,\tfrac{1}{2}],0,0), \nonumber \end{eqnarray} by Lemma \ref{traceProjections}. Since $S_2(1)$ contains the correlations $(1/2,0,0), (0,1/2,0)$, and $(0,0,1/2)$, we get $C = \tfrac{2}{3} \co\{\vec{0}, S_2(1)\}.$ \end{proof}


We will now begin to analyze the structure of $S_d(n)$ for any $n < d/2$. We start by applying Lemma \ref{mainLemma} to this particular setting.

\begin{lemma} \label{lemmaTypeI} Suppose that $n,m$ and $d$ are positive integers such that $2n = d -m$ and set $l = \min(n-1,m)$. Then \[ S_d(n) \subseteq \co\{ \tfrac{d-1}{d} S_{d-1}(n), \tfrac{d-3l}{d} S_{d-3l}(n-l)\}. \] \end{lemma}

\begin{proof} By repeated application of Lemma \ref{mainLemma} we have
\begin{eqnarray} S_d(n) & \subseteq & \tfrac{d-1}{d} \co\{ S_{d-1}(n), S_{d-1}(n-1,n,n)\}, \nonumber \\
\tfrac{d-1}{d}S_{d-1}(n-1,n,n) & \subseteq & \tfrac{d-2}{d} \co\{ S_{d-2}(n-1,n,n), S_{d-2}(n-1, n-1, n)\}, \nonumber \\
\tfrac{d-2}{d} S_{d-2}(n-1, n-1, n) & \subseteq & \tfrac{d-3}{d} \co\{  S_{d-3}(n -1, n-1, n), S_{d-3}(n-1)\}, \nonumber \\
\tfrac{d-3}{d} S_{d-3}(n-1) & \subseteq & \tfrac{d-4}{d} \co \{ S_{d-4}(n-1), S_{d-4}(n-2, n-1, n-1) \} \\
\tfrac{d-4}{d} S_{d-4}(n-2,n-1,n-1) & \subseteq & \tfrac{d-5}{d} \co \{ S_{d-5}(n-2,n-1,n-1), S_{d-5}(n-2, n-2, n-1) \} \nonumber \\
\tfrac{d-5}{d} S_{d-5}(n-2,n-2,n-1) & \subseteq & \tfrac{d-6}{d} \co \{ S_{d-6}(n-2,n-2,n-1), S_{d-6}(n-2) \} \nonumber \\
& \dots & \nonumber \\
\tfrac{d-3l+3}{d} S_{d-3l + 3}(n-l+1) & \subseteq & \tfrac{d-3l + 2}{d} \co\{ S_{d-3l + 2}(n-l+1), \nonumber \\ & &  S_{d-3l + 2}(n-l,n-l+1,n-l+1)\}, \nonumber \\
\tfrac{d-3l + 2}{d}S_{d-3l + 2}(n-l,n-l+1,n-l+1) & \subseteq & \tfrac{d-3l + 1}{d} \co\{ S_{d-3l + 1}(n-l,n-l+1,n-l+1), \nonumber \\ & & \qquad S_{d-2}(n-l, n-l, n-l+1)\}, \nonumber \\
\tfrac{d-3l + 1}{d} S_{d-3l + 1}(n-l, n-l, n-l+1) & \subseteq & \tfrac{d-3l}{d} \co\{  S_{d-3l}(n -l, n-l, n-l+1), S_{d-3l}(n-l)\}.  \end{eqnarray}
Notice that equation (3) follows from Lemma \ref{mainLemma} and the observation that $2n-2 < d-3$ if and only if $2n < d - 1$ if and only if $l > 1$ (if $l=1$, then the sequence terminates at the line preceding (3)). Similar observations produce the entire sequence of inclusions. The sequence terminates at equation (4) where either $l=m$ so that $2(n-l)=2n - 2m = d-3l$ or $n-l = 1$. It follows that \[ S_d(n) \subseteq \co \{ \tfrac{d-1}{d} S_{d-1}(n), \tfrac{d-2}{d} S_{d-2}(n-1,n,n), \tfrac{d-3}{d}S_{d-3}(n-1,n-1,n), \tfrac{d-4}{d} S_{d-4}(n-1), \dots, \] \[ \quad \tfrac{d-3l}{d}S_{d-3l}(n-l, n-l, n-l+1), \tfrac{d-3l}{d} S_{d-3l}(n-l) \}. \] Furthermore, by Lemma \ref{LemmaBlowUp} we have
\begin{eqnarray} \tfrac{d-3l}{d} S_{d-3l}(n-l, n-l, n-l+1) & \subseteq & \tfrac{d-3l + 1}{d} S_{d-3l + 1}(n-l,n-l+1,n-l+1) \nonumber \\ & \subseteq & \tfrac{d-3l + 2}{d} S_{d-3l+2}(n-l+1) \nonumber \\ & \dots & \nonumber \\
& \subseteq & \tfrac{d-3}{d}S_{d-3}(n-1,n-1,n) \nonumber \\ & \subseteq & \tfrac{d-2}{d} S_{d-2}(n-1,n,n) \nonumber \\ & \subseteq & \tfrac{d-1}{d} S_{d-1}(n). \nonumber \end{eqnarray} The statement follows by combining first chain of inclusions with the second. \end{proof}

\begin{proposition} \label{propTypeI} Suppose that $n$ and $d$ are positive integers such that $2n \leq d$. Then \[ S_d(n) \subseteq \co\{ \max(0, \tfrac{6n - 2d}{d})S_2(1), \tfrac{2n}{d} S_2(1)\}. \] \end{proposition}

\begin{proof} First notice that if $2n = d$ the statement is true by Proposition \ref{PropHalfSlice} since $\tfrac{6n-2d}{d} = \tfrac{2n}{d} = 1$ in that case. Otherwise $2n = d - m$ for some positive $m$. For each $j \in \{ 0, 1, 2, \dots, m \}$, define $l_j = \min(n-1, m-j)$ and $d_j = d - j - 3l_j$. By Lemma \ref{lemmaTypeI}, we have the following inclusions. \begin{eqnarray} S_d(n) & \subseteq & \co\{ \tfrac{d_0}{d} S_{d_0}(n-l_0), \tfrac{d-1}{d} S_{d-1}(n)\} \nonumber \\
\tfrac{d-1}{d} S_{d-1}(n) & \subseteq & \co\{ \tfrac{d_1}{d} S_{d_1}(n-l_1), \tfrac{d-2}{d} S_{d-2}(n)\} \nonumber \\
& \dots & \nonumber \\
\tfrac{d-m+1}{d} S_{d-m+1}(n) & \subseteq & \co\{ \tfrac{d_{m-1}}{d} S_{d_{m-1}}(n-l_{m-1}), \tfrac{d_m}{d} S_{d_m}(n)\}. \nonumber \end{eqnarray} Combining the above chain of inclusions, we get \[ S_d(n) \subseteq \co\{ \tfrac{d_0}{d} S_{d_0}(n-l_0), \tfrac{d_1}{d} S_{d_1}(n-l_1), \dots, \tfrac{d_{m-1}}{d} S_{d_{m-1}}(n-l_{m-1}), \tfrac{d_m}{d} S_{d_m}(n) \}. \]

We must consider two cases: when $n \leq m$ and when $n > m$. We begin with the case $n \leq m$. Then for $j = 0, 1, 2, \dots, m-n$ we have $l_j = min(n-1, m-j) = n-1$. By Proposition \ref{PropEasyCase} we have \begin{eqnarray} \tfrac{d_j}{d} S_{d_j}(n - l_j) & = & \tfrac{d_j}{d} S_{d_j}(1) \nonumber \\
& \subseteq & \co\{ \vec{0}, \tfrac{2}{d} S_2(1)\}. \nonumber \end{eqnarray} On the other hand, for each $j = m-n+1, m-n+2, \dots, m$ we have $l_j = m-j$, so that \begin{eqnarray} d_j & = & d +2j - 3m \nonumber \\ & = & (d-m) - 2m + 2j \nonumber \\ & = & 2n + 2j - 2m \\ & = & 2(n-l_j). \nonumber \end{eqnarray} From Proposition \ref{PropHalfSlice} we get $S_{d_j}(n-l_j)=S_2(1)$ in this case. Thus we have \[ \tfrac{d_j}{d} S_{d_j}(n-l_j) = \tfrac{2(n-m+j)}{d} S_2(1). \] It follows that \[ S_d(n) \subseteq \co\{ \vec{0}, \tfrac{2}{d} S_2(1), \tfrac{4}{d} S_2(1), \dots, \tfrac{2n}{d} S_2(1)\} \subseteq \co\{ \vec{0}, \tfrac{2n}{d} S_2(1)\}. \]

In the case when $n > m$, we simply observe that $l_0 = m$ implies that $2(n-l_0) = 2(n-m) = 6n-2d$. Repeating the arguments of the previous case, we get \begin{eqnarray} S_d(n) & \subseteq & \co\{ \tfrac{6n-2d}{d} S_2(1), \tfrac{6n-2d + 2}{d} S_2(1), \dots, \tfrac{2n}{d} S_2(1)\} \nonumber \\ & \subseteq & \co\{ \tfrac{6n-2d}{d} S_2(1), \tfrac{2n}{d} S_2(1)\} \nonumber \end{eqnarray} concluding the proof. \end{proof}

\begin{remark} \label{typeIRemark} \emph{We will see later that Proposition \ref{propTypeI} essentially characterizes the $\vec{r}$-slices of $C_{max}^s(3,2)$ for $\vec{r}=(r_0,r_0,r_0)$ with $r_0$ rational. We outline the argument here. Indeed, the sets $\max(0, \tfrac{6n - 2d}{d})S_2(1)$ and $\tfrac{2n}{d} S_2(1)$ are actually subsets of $S_d(n)$. To see this, we just need to demonstrate $d \times d$ projections of rank $n$ implementing the correlations represented by these sets. This can be done for $\tfrac{2n}{d} S_2(1) = \tfrac{2n}{d} S_{2n}(n)$ with projections of the form \begin{equation} \label{t1rep1} Q_i = (\hat{Q}_i \otimes I_n) \oplus 0_{d-2n} \end{equation} for $i=1,2,3$, where the $\hat{Q}_i$'s are $2 \times 2$ rank one projections. For $\max(0, \tfrac{6n - 2d}{d})S_2(1)$, we must consider the cases when $3n \leq d$ and $3n > d$ separately. When $3n \leq d$, this becomes the singleton set containing the correlation $(0,0,0)$. This correlation is realized with projections of the form \begin{equation} \label{t1rep2} P_i = \delta_{i,1} I_n \oplus \delta_{i,2} I_n \oplus \delta_{i,3} I_n \oplus 0_{d-3n}. \end{equation} In the other case, we may realize $\tfrac{6n - 2d}{d} S_2(1)$ using projections of the form \begin{equation} \label{t1rep3} P_i = (\hat{P}_i \otimes I_{3n-d}) \oplus I_{d-2n} \delta_{i,1} \oplus I_{d-2n} \delta_{i,2} \oplus I_{d-2n} \delta_{i,3} \end{equation} where again $\hat{P}_i$ denotes a $2 \times 2$ rank one projection. We can then build any correlation in the convex hull $\co\{ \max(0, \tfrac{6n - 2d}{d})S_2(1), \tfrac{2n}{d} S_2(1)\}$ by considering arbitrary traces on $\mathbb{M}_d \oplus \mathbb{M}_d$ and projections of the form $R_i = P_i \oplus Q_i$.} \end{remark}

The representation of $\co\{ \max(0, \tfrac{6n - 2d}{d})S_2(1), \tfrac{2n}{d} S_2(1)\}$ discussed above can also be used to characterize slices of $C_{max}^s(3,2)$ which are obtained from type I slices by swapping experiments or reversing outcomes. We mention one case here, which we will need to understand type II correlations.

\begin{lemma} \label{typeISwap} Suppose that $2n \leq d$. Then \[ S_d(n,d-n,d-n) \subseteq \co\{\max(0,\tfrac{6n-2d}{d}) S_2(1) + \vec{a}, \tfrac{2n}{d} S_2(1) + \vec{b}\} \] where $\vec{a} = (\tfrac{\min(n,d-2n)}{d},\tfrac{\min(n,d-2n)}{d},\tfrac{d-2n}{d})$ and $\vec{b} = (0,0,\tfrac{d-2n}{d})$. \end{lemma}

\begin{proof} Assume $P_1, P_2$, and $P_3$ are $d \times d$ projections of rank $n$, $d-n$ and $d-n$, respectively. Then $Q_1=P_1, Q_2 = I - P_2,$ and $Q_3 = I - P_3$ are projections of rank $n$. We proceed by considering the cases when the vector $p=(\tr_d(Q_1 Q_2), \tr_d(Q_1 Q_3), \tr_d(Q_2 Q_3))$ is an element of $\max(0,\tfrac{6n-2d}{d}) S_2(1)$ or $\tfrac{2n}{d} S_2(1)$ from which the general case will follow.

First, assume $p \in \max(0,\tfrac{6n-2d}{d}) S_2(1)$. As shown in Remark \ref{typeIRemark}, we have \[ p = (\tr_d(Q'_1 Q'_2), \tr_d(Q'_1 Q'_3), \tr_d(Q'_2 Q'_3)) \] where $Q'_i$ has the form given in equation ($\ref{t1rep2}$) if $3n \leq d$ or ($\ref{t1rep3}$) otherwise. Set $P'_i = I - Q'_i$ for $i=\{2,3\}$ and $P'_1 = Q'_1$. First assume $3n > d$. Then each $P'_i$ has the form \[ P'_i = (\hat{P}_i \otimes I_{3n-d}) \oplus (I_{d-2n} \oplus I_{d-2n} \delta_{i,3} \oplus I_{d-2n} \delta_{i,2}). \] Using $\tr_d(A \oplus B) = \tfrac{d_1}{d} \tr_{d_1}(A) + \tfrac{d_2}{d}\tr_{d_2}(B)$ for $d_1 = 6n - 2d$ and $d_2 = 3d-6n$, we observe that \begin{eqnarray} (\tr_d(P'_1 P'_2), \tr_d(P'_1 P'_3), \tr_d(P'_2 P'_3)) & \in & \tfrac{6n-2d}{d} S_2(1) + (\tfrac{d-2n}{d}, \tfrac{d-2n}{d}, \tfrac{d-2n}{d}). \nonumber \end{eqnarray} When $3n \leq d$, the $P'_i$ have the form \[ P'_i = I_n \oplus \delta_{i,3} I_n \oplus \delta_{i,2} I_n \oplus (\delta_{i,2} + \delta_{i,3}) I_{d-3n} \] from which it follows that \begin{eqnarray} (\tr_d(P'_1 P'_2), \tr_d(P'_1 P'_3), \tr_d(P'_2 P'_3)) & = & (\tfrac{n}{d}, \tfrac{n}{d}, \tfrac{d-2n}{d}). \nonumber \end{eqnarray} Finally, if we assume that $p \in \tfrac{2n}{d} S_2(1)$, then we have \[ p = (\tr_d(Q'_1 Q'_2), \tr_d(Q'_1 Q'_3), \tr_d(Q'_2 Q'_3)) \] with $Q'_i$ having the form given in equation ($\ref{t1rep1}$) above. It follows that for $P'_1 = Q'_1, P'_i = I - Q'_i$ ($i=2,3$) we have \[ P'_i = (\hat{Q}_i \otimes I_n) \oplus (\delta_{i,2} + \delta_{i,3}) I_{d-2n} \] and thus \begin{eqnarray} (\tr_d(P'_1 P'_2), \tr_d(P'_1 P'_3), \tr_d(P'_2 P'_3)) & \in & \tfrac{2n}{d} S_2(1) + (0,0, \tfrac{d-2n}{d}). \nonumber \end{eqnarray} The general result follows by considering convex combinations of the above cases. \end{proof}

\subsection{Type II slices}

We are ready to consider the type II slices of $C_{max}^s(3,2)$. Our strategy will be to show that this set can be described in terms of the type I slices and in terms of the $(r,1-r,1-r)$-slices considered in Lemma \ref{typeISwap}.

The next lemma applies Lemma \ref{mainLemma} to the setting of type II slices. It will allow us to describe arbitrary type II slices in terms of type I slices and the swapped type I slices of Proposition \ref{typeISwap}. Since the proof is similar to the proof of Lemma \ref{lemmaTypeI} we leave some details to the reader.

\begin{lemma} \label{lemmaTypeII} Suppose that $k,n,m$ and $d$ are positive integers with $2n+k = d - m$ and set $l = \min(k,m)$. Then \[ S_d(n,n+k,n+k) \subseteq \co\{ \tfrac{d-1}{d} S_{d-1}(n,n+k,n+k), \tfrac{d-2l}{d} S_{d-2l}(n, n+k-l, n+k-l)\}. \] \end{lemma}

\begin{proof} By repeated application of Lemma \ref{mainLemma} we have \begin{eqnarray} S_d(n,n+k,n+k) & \subseteq & \tfrac{d-1}{d} \co \{ S_{d-1}(n,n+k,n+k), \nonumber \\
& & \qquad \qquad S_{d-1}(n,n+k-1,n+k) \}, \nonumber \\
\tfrac{d-1}{d}S_{d-1}(n,n+k-1,n+k) & \subseteq & \tfrac{d-2}{d} \co \{ S_{d-2}(n, n+k-1, n+k), \nonumber \\
& & \qquad \qquad S_{d-2}(n, n+k-1, n+k-1) \}, \nonumber \\
\tfrac{d-2}{d} S_{d-2}(n, n+k-1, n+k-1) & \subseteq & \tfrac{d-3}{d} \co \{  S_{d-3}(n, n+k-1,n+k-1), \nonumber \\
& & \qquad \qquad S_{d-3}(n,n+k-2,n+k-1) \}, \nonumber \\
\tfrac{d-3}{d} S_{d-3}(n,n+k-2,n+k-1) & \subseteq & \tfrac{d-4}{d} \co \{  S_{d-4}(n, n+k-2,n+k-1), \nonumber \\
& & \qquad \qquad S_{d-4}(n,n+k-2,n+k-2) \}, \nonumber \\
& \dots & \nonumber \end{eqnarray}
\begin{eqnarray} \tfrac{d-2l+2}{d} S_{d-2l+2}(n,n+k-l+1,n+k-l+1) & \subseteq & \tfrac{d-2l + 1}{d} \co \{ S_{d-2l + 1}(n,n+k-l+1,n+k-l+1), \nonumber \\
& & \qquad \qquad S_{d-2l + 1}(n, n+k-l, n+k-l+1) \}, \nonumber \\
\tfrac{d-2l + 1}{d} S_{d-2l + 1}(n, n+k-l, n+k-l+1) & \subseteq & \tfrac{d-2l}{d} \co \{ S_{d-2l}(n, n+k-l, n+k-l+1), \nonumber \\
& & \qquad \qquad S_{d-2}(n, n+k-l, n+k-l) \}. \nonumber \end{eqnarray} Furthermore, by Lemma \ref{LemmaBlowUp} we have \begin{eqnarray} \tfrac{d-2l}{d} S_{d-2l}(n, n+k-l, n+k-l+1) & \subseteq & \tfrac{d-2l + 1}{d} S_{d-2l + 1}(n,n+k-l+1,n+k-l+1) \nonumber \\
& \subseteq & \tfrac{d-2l + 2}{d} S_{d-2l+2}(n,n+k-l+1,n+k-l+2) \nonumber \\
& \dots & \nonumber \\ & \subseteq & \tfrac{d-3}{d} S_{d-3}(n, n+k-1,n+k-1) \nonumber \\
& \subseteq & \tfrac{d-2}{d} S_{d-2}(n, n+k-1, n+k) \nonumber \\
& \subseteq & \tfrac{d-1}{d} S_{d-1}(n,n+k,n+k). \nonumber \end{eqnarray} The statement follows by combining first chain of inclusions with the second. \end{proof}

We are now ready to characterize type II slices.

\begin{proposition} \label{propTypeII} Suppose that $n,k$ and $d$ are positive integers with $2(n + k) \leq d$. Define convex sets \[ A_1 := \max(0, \tfrac{6n+4k-2d}{d})S_2(1), \qquad A_2:= \tfrac{2n}{d} S_2(1) + (0,0,[0,\tfrac{k}{d}]). \] Then $S_d(n,n+k,n+k) \subseteq \co\{ A_1, A_2 \}$. \end{proposition}

\begin{proof} Assume $2n + k = d - m$. Since $2(n+k) \leq d$, we see that $m \geq k$. For each $j \in \{ 0,1,2, \dots, m\}$, define $l_j = \min(k,m-j)$ and $d_j = d - j - 2l_j$, and $k_j = k - l_j$. By Lemma \ref{lemmaTypeII}, we have the following inclusions.
\begin{eqnarray}
S_d(n,n+k,n+k) & \subseteq & \co\{ \tfrac{d-1}{d} S_{d-1}(n,n+k,n+k), \nonumber \\
& & \qquad \qquad \tfrac{d_0}{d} S_{d_0}(n, n+k_0, n+k_0)\} \nonumber \\
\tfrac{d-1}{d} S_{d-1}(n,n+k,n+k) & \subseteq & \co\{ \tfrac{d-2}{d} S_{d-2}(n,n+k,n+k), \nonumber \\
& & \qquad \qquad \tfrac{d_1}{d} S_{d_1}(n, n+k_1, n+k_1)\} \nonumber \\
& \dots & \nonumber \\
\tfrac{d-m+1}{d} S_{d-m+1}(n,n+k,n+k) & \subseteq & \co\{ \tfrac{d_m}{d} S_{d_m}(n,n+k,n+k), \nonumber \\
& & \qquad \tfrac{d_{m-1}}{d} S_{d_{m-1}}(n, n+k_{m-1}, n+k_{m-1}) \}. \nonumber \end{eqnarray} Combining the above inclusions gives us \begin{eqnarray} S_d(n,n+k,n+k) & \subseteq & \co\{ \tfrac{d_0}{d} S_{d_0}(n,n+k_0,n+k_0), \tfrac{d_1}{d} S_{d_1}(n,n+k_1,n+k_1), \dots \nonumber \\
& & \tfrac{d_{m-1}}{d} S_{d_{m-1}}(n,n+k_{m-1},n+k_{m-1}), \tfrac{d_m}{d} S_{d_m}(n,n+k,n+k)\}. \nonumber \end{eqnarray}

Assume $m = k + m'$. Then for each $j \in \{0,1,\dots, m'\}$ we have $l_j = \min(k,m-j) = k$, and hence $k_j = 0$ for $j \leq m'$. For $j > m'$, $l_j = m-j$, so that \begin{eqnarray} 2n + k_j & = & 2n + k + j - m \nonumber \\
& = & d - 2m + j \nonumber \\
& = & d_j. \nonumber \end{eqnarray} Obviously $2n+k_j = d_j$ if and only if $n + k_j = d_j - n$. It follows that for each $j > m'$ \[ \tfrac{d_j}{d} S_{d_j}(n,n+ k_j, n+k_j) = \tfrac{d_j}{d} \pi_{2,3}S_{d_j}(n) \] where $\pi_{2,3} = \pi_2 \circ \pi_3$ denotes the the affine map given by reversing the outcomes of the second experiment and the third experiment. Hence \begin{eqnarray} S_d(n,n+k,n+k) & \subseteq & \co\{ \tfrac{d_0}{d} S_{d_0}(n), \tfrac{d_1}{d} S_{d_1}(n), \dots, \tfrac{d_{m'}}{d} S_{d_{m'}}(n), \nonumber \\& & \tfrac{d_{m'+1}}{d} \pi_{2,3} S_{d_{m'+1}}(n), \dots, \tfrac{d_m}{d} \pi_{2,3} S_{d_m}(n)\}.  \end{eqnarray}

We need to show that for each $j \leq m'$ we have \[ \tfrac{d_j}{d} S_{d_j}(n) \subseteq \co\{A_1, A_2\} \] and that for each $j > m'$ we have \[ \tfrac{d_j}{d} \pi_{2,3} S_{d_j}(n) \subseteq \co\{A_1, A_2\}. \] In the case when $j \leq m'$, we have $d_j = d-j-2k$. Hence \[ \tfrac{d_j}{d} S_{d_j}(n) \subseteq \co\{ \max(0, \tfrac{6n + 4k + 2j -2d}{d})S_2(1), \tfrac{2n}{d} S_2(1)\} \] by Proposition \ref{propTypeI} and the equality $6n-2d_j = 6n + 4k + 2j - 2d$. Since $m' = m - k$, we see that
\begin{eqnarray} 6n + 4k - 2d & \leq & 6n + 4k + 2j - 2d \nonumber \\
& \leq & 6n + 4k + 2(m-k) - 2d \nonumber \\
& = & 6n + 2k + 2m - 2d \nonumber \\
& = & 2n + 2(2n+k) + 2m - 2d \nonumber \\
& = & 2n + 2(d-m) + 2m - 2d \nonumber \\
& = & 2n \nonumber \end{eqnarray} for each $j \leq m'$. So for each $j \leq m'$, $6n + 4k - 2d \leq 6n +4k +2j - 2d \leq 2n$, which implies \begin{eqnarray} \tfrac{d_j}{d}S_{d_j}(n) & \subseteq & \co\{ \max(0, \tfrac{6n + 4k -2d}{d}) S_2(1), \tfrac{2n}{d} S_2(1)\} \nonumber \\ & \subseteq & \co\{A_1, A_2\}. \nonumber \end{eqnarray}

Lastly we consider the case $j > m'$. By Lemma \ref{typeISwap}, we have \[ \tfrac{d_j}{d} \pi_{2,3} S_{d_j}(n) \subseteq \co\{ \max(0, \tfrac{6n-2d_j}{d}) S_2(1) + \vec{a}_j, \tfrac{2n}{d} S_2(1) + \vec{b}_j\} \] where $\vec{a}_j = (\tfrac{\min(n,d_j - 2n)}{d}, \tfrac{\min(n,d_j - 2n)}{d}, \tfrac{d_j - 2n}{d})$ and $\vec{b}_j = (0, 0, \tfrac{d_j - 2n}{d})$. We will show that \[ \co\{ \max(0, \tfrac{6n-2d_j}{d}) S_2(1) + \vec{a}_j, \tfrac{2n}{d} S_2(1) + \vec{b}_j\} \subseteq A_2. \] Let $p \in \max(0,\tfrac{6n-2d_j}{d}) S_2(1) + \vec{a}_j$. If $d_j < 3n$ then $6n - 2d_j > 0$ so that\[ p \in \tfrac{6n-2d_j}{d}S_2(1) + (\tfrac{d_j-2n}{d},\tfrac{d_j-2n}{d},\tfrac{d_j-2n}{d}).\] Thus $ \tfrac{d}{2n} p = (\tr_{2n}(P_1 P_2), \tr_{2n}(P_1 P_3), \tr_{2n}(P_2 P_3) )$ where $P_i = (\hat{P}_i \otimes I_{3n-d_j}) \oplus (I_{d_j - 2n} \oplus 0_{d_j - 2n})$ with $\hat{P}_i$ a $2 \times 2$ projection. Notice that $P_i$ is a rank $n$ matrix of size $2n \times 2n$. Hence $\tfrac{d}{2n}p \in S_2(1)$ by Proposition \ref{PropHalfSlice}. It follows that $p \in \tfrac{2n}{d} S_2(1) \subseteq A_2$. On the other hand, if $d_j \geq 3n$, then $6n - 2d_j \leq 0$ and $p = \vec{a}_j$. Using projections of the form $P_i =  I_{n} \oplus 0_{n}$ we see that $(\tfrac{n}{d},\tfrac{n}{d},\tfrac{n}{d}) \in \tfrac{2n}{d} S_2(1)$. Since $d_j = 2n + k_j \leq 2n + k$, we have $d_j - 3n \leq k - n \leq k$. It follows that $p=\vec{a}_j = (\tfrac{n}{d},\tfrac{n}{d},\tfrac{n}{d}) + (0,0, \tfrac{d_j-3n}{d}) \in \tfrac{2n}{d} S_2(1) + (0,0,[0,\tfrac{k}{d}]) = A_2$. Finally, in either case we have $d_j - 2n = k_j \leq k$, and hence \[ \tfrac{2n}{d} S_2(1) + \vec{b}_j \subseteq A_2. \] We conclude that for each $j > m'$, \[ \tfrac{d_j}{d} \pi_{2,3}S_{d_j}(n) \subseteq A_2  \] and the proposition follows. \end{proof}

\begin{remark} \label{typeIIRemark} \emph{Similar to the case of type I correlations, it turns out that the sets $A_1$ and $A_2$ from Proposition \ref{propTypeII} are actually subsets of $S_d(n,n+k,n+k)$. Indeed, the set $A_2 = \tfrac{2n}{d} S_2(1) + (0,0,[0,\tfrac{k}{d}])$ can be implemented with projections of the form \begin{equation} \label{t2rep1} P_i = (\hat{P}_i \otimes I_n) \oplus (\delta_{i,2}+\delta_{i,3})(\hat{Q}_i \otimes I_k) \oplus 0_{d-2n-2k} \end{equation} for $2 \times 2$ rank-1 projections $\hat{P}_i$ and $\hat{Q}_i$. When $3n + 2k - d \leq 0$, we have \[ A_1 = \max(0, \tfrac{6n+4k-2d}{d})S_2(1) = (0,0,0), \] which can be implemented with projections of the form \begin{equation} \label{t2rep2} Q_i = \delta_{i,1} I_n \oplus \delta_{i,2} I_{n+k} \oplus \delta_{i,3} I_{n+k} \oplus 0_{d-3n-2k}. \end{equation} Otherwise, $A_1 = \tfrac{6n+4k-2d}{d} S_2(1)$ and can be implemented with projections of the form \begin{equation} \label{t2rep3} Q_i = (\hat{R}_i \otimes I_{3n+2k-d}) \oplus \delta_{i,1} I_{d-2n-2k} \oplus \delta_{i,2} I_{d-2n-k} \oplus \delta_{i,3} I_{d-2n-k}. \end{equation} To implement correlations in the convex hull of these sets, one can consider direct sums of the projections above together with arbitrary traces on $\mathbb{M}_d \oplus \mathbb{M}_d$.} \end{remark}

As with type I slices, the representations provided in Remark \ref{typeIIRemark} allow us to find expressions for correlations obtained from type II correlations by swapping experiments or reversing outcomes. We record one special case here, which will help us determine the structure of the type III slices.

\begin{lemma} \label{propTypeIISwap} Suppose that $2(n + k) \leq d$. Then \[ S_d(n, n + k, d - (n+k)) \subseteq \co\{ \max(0, \tfrac{6n + 4k - 2d}{d}) S_2(1) + \vec{a}, \tfrac{2n}{d} S_2(1) + (0,0,[0,\tfrac{k}{d}]) \} \] where $\vec{a} = (0, \tfrac{\min(n, d-2n-2k)}{d}, \tfrac{\min(n,d-2n-2k) + k}{d}).$ \end{lemma}

\begin{proof} By Remark \ref{typeIIRemark}, we may implement any correlation in $A_2$ with projections of the form $P_i$ given in (\ref{t2rep1}) and any correlation in $A_1$ with projections of the form $Q_i$ given in (\ref{t2rep2}) or (\ref{t2rep3}). We will show that replacing $P_3$ with $I-P_3$ yields a correlation in $\tfrac{2n}{d} S_2(1) + (0,0,[0,\tfrac{k}{d}]) = A_2$ and replacing $Q_3$ with $I-Q_3$ yields a correlation in $\max(0, \tfrac{6n + 4k - 2d}{d}) S_2(1) + \vec{a}$. The Lemma follows from these observations and Proposition \ref{propTypeII}.

First, observe that $I-P_3$ has the form \[ I-P_3 = ((I_2 - \hat{P}_i) \otimes I_n) \oplus (\delta_{i,2}+\delta_{i,3})((I-\hat{Q}_i) \otimes I_k) \oplus I_{d-2n-2k}. \] It follows that \[ (Tr(P_1 P_2), Tr(P_1 (I-P_3)), Tr(P_2(I-P_3))) = (Tr(P_1 P_2), Tr(P_1 R_3), Tr(P_2, R_3)) \] where \[ R_3 = ((I_2 - \hat{P}_i) \otimes I_n) \oplus ((I-\hat{Q}_i) \otimes I_k) \oplus 0_{d-2n-2k} \] since $P_1$ and $P_2$ each have the form given by (\ref{t2rep1}) and hence have a final summand of $0_{d-2n-2k}$. It follows that replacing $P_3$ by $I-P_3$ yields another correlation in $A_2$.

Next, observe that $I-Q_3$ has the form \[ I - Q_3 = I_n \oplus I_{n+k} \oplus 0_{n+k} \oplus I_{d - 3n - 2k} \] when $3n+2k-d \leq 0$ and \[ I - Q_3 = ( (I-\hat{R}_3) \otimes I_{3n + 2k - d}) \oplus (I_{d-2n-2k} \oplus I_{d-2n-k} \oplus 0_{d-2n-k}) \] otherwise. In the first case we get \[ \tfrac{1}{d}(Tr(Q_1 Q_2), Tr(Q_1 (I-Q_3)), Tr(Q_2 (I-Q_3))) = (0,\tfrac{n}{d},\tfrac{n+k}{d}) = \vec{a} \] where the last equality follows since $n \leq d-2n-2k$. In the second case, we may use $\tr_d(A \oplus B) = \tfrac{d_1}{d} \tr_{d_1} Tr(A) + \tfrac{d_2}{d} \tr_{d_2}(B)$ to get  \[ \tfrac{1}{d}(Tr(Q_1 Q_2), Tr(Q_1 (I-Q_3)), Tr(Q_2 (I-Q_3))) \in \tfrac{6n+4k-2d}{d} S_2(1) + (0, \tfrac{d-2n-2k}{d}, \tfrac{d-2n-k}{d}) \] completing the proof. \end{proof}


\subsection{Type III slices}

Finally we must consider the type III slices of $C_{max}^s(3,2)$. Here our strategy will be to show that type III slices can be described in terms of type II slices and the $(r, r', 1-r')$-slices produced in Lemma \ref{propTypeIISwap}. We begin by applying Lemma \ref{mainLemma} to the setting of type III correlations. The proof is similar to the proofs of Lemma \ref{lemmaTypeI} and Lemma \ref{lemmaTypeII} so we leave some details to the reader.

\begin{lemma} \label{lemmaTypeIII} Suppose that $k,k',n,m$ and $d$ are positive integers such that $k \leq k'$ and $2n+k+k' = d - m$ and set $l = \min(n,m)$. Then \[ S_d(n,n+k,n+k') \subseteq \co\{ \tfrac{d-1}{d} S_{d-1}(n,n+k,n+k'), \tfrac{d-l}{d} S_{d-l}(n-l,n+k,n+k')\}. \] \end{lemma}

\begin{proof} By repeated application of Lemma \ref{mainLemma} we have \begin{eqnarray} S_d(n, n+k, n+k') & \subseteq & \tfrac{d-1}{d} \co \{ S_{d-1}(n,n+k,n+k'), \nonumber \\
& & \qquad \qquad S_{d-1}(n-1,n+k,n+k) \}, \nonumber \\
\tfrac{d-1}{d}S_{d-1}(n-1,n+k,n+k') & \subseteq & \tfrac{d-2}{d} \co \{ S_{d-2}(n-1, n+k, n+k'), \nonumber \\
& & \qquad \qquad S_{d-2}(n-2, n+k, n+k') \}, \nonumber \\
\tfrac{d-2}{d}S_{d-2}(n-2,n+k,n+k') & \subseteq & \tfrac{d-3}{d} \co \{ S_{d-3}(n-2, n+k, n+k'), \nonumber \\
& & \qquad \qquad S_{d-3}(n-3, n+k, n+k') \}, \nonumber \\
& \dots & \nonumber \\
\tfrac{d-l+2}{d}S_{d-l+2}(n-l+2,n+k,n+k') & \subseteq & \tfrac{d-l+1}{d} \co \{ S_{d-l+1}(n-l+2, n+k, n+k'), \nonumber \\
& & \qquad \qquad S_{d-l+1}(n-l+1, n+k, n+k') \}, \nonumber \\
\tfrac{d-l+1}{d}S_{d-2}(n-l+1,n+k,n+k') & \subseteq & \tfrac{d-l}{d} \co \{ S_{d-l}(n-l+1, n+k, n+k'), \nonumber \\
& & \qquad \qquad S_{d-l}(n-l, n+k, n+k') \}. \nonumber \end{eqnarray} Furthermore, by Lemma \ref{LemmaBlowUp} we have \begin{eqnarray} \tfrac{d-l}{d} S_{d-l}(n-l+1, n+k, n+k') & \subseteq & \tfrac{d-l + 1}{d} S_{d-l + 1}(n-l+2,n+k+1,n+k') \nonumber \\
& \subseteq & \tfrac{d-l + 2}{d} S_{d-l+2}(n-l+2,n+k,n+k') \nonumber \\
& \dots & \nonumber \\
& \subseteq & \tfrac{d-2}{d} S_{d-2}(n-1, n+k, n+k') \nonumber \\
& \subseteq & \tfrac{d-1}{d} S_{d-1}(n,n+k,n+k'). \nonumber \end{eqnarray} The statement follows by combining first chain of inclusions with the second.  \end{proof}

\begin{proposition} \label{propositionTypeIII} Suppose that $n,k,k'$ and $d$ are positive integers such that $2(n+k') \leq d$ and $k \leq k'$. Then $S_d(n,n+k,n+k')$ is contained in the convex hull of the sets \[ B_1 := \max(0,\tfrac{6n+2k+2k'-2d}{d})S_2(1), \qquad B_2 := \tfrac{2n}{d} S_2(1) + (0,0,[0,\tfrac{k}{d}]), \] and \[ B_3 := \max(0,\tfrac{2(n+k-k')}{d}) S_2(1) + (0, \tfrac{\min(n,k'-k)}{d}, \tfrac{\min(n,k'-k)+k}{d}). \] \end{proposition}

\begin{proof} First, if $k=k'$ and $2(n+k)=d$, then the statement follows from the observation that $B_1 =A_1, B_2 = A_2$, and $B_3 \subseteq A_2$, where $A_1$ and $A_2$ are the sets described in Proposition \ref{propTypeII}. Otherwise we may assume $2n + k + k' = d - m$ for some positive integer $m$. For each $j = 0,1,2,\dots,m$, define $l_j = \min(n, m-j)$, $d_j = d - j - l_j$, and $n_j = n - l_j$. Then, by repeated application of Lemma \ref{lemmaTypeIII}, we have \begin{eqnarray} S_d(n, n+k, n+k') & \subseteq & \co\{ \tfrac{d_0}{d} S_{d_0}(n_0,n+k,n+k'), \tfrac{d_1}{d} S_{d_1}(n_1,n+k,n+k'), \dots \nonumber \\
& & \dots, \tfrac{d_{m-1}}{d} S_{d_{m-1}}(n_{m-1}, n+k, n+k'), \tfrac{d_m}{d} S_{d_m}(n_m, n+k, n+k') \}. \nonumber \end{eqnarray}

We will show that each set \[\tfrac{d_j}{d} S_{d_j}(n_j,n+k,n+k')\] is contained in the convex hull of the three sets $B_1, B_2,$ and $B_3$. We will need to consider two cases, depending on whether or not $j \leq m -n$.

We begin by considering the case when $j \leq m - n$. Obviously we only need to consider this case only when $n \leq m$. When this occurs, we have $l_j = n$ and hence $n_j = 0$ and $d_j = d - j - n$. Thus \[ \tfrac{d_j}{d} S_{d_j}(n_j, n+k, n+k') = \tfrac{d-j-n}{d} S_{d-j-n}(0, n+k, n+k'). \] We can characterize this set using Lemma \ref{traceProjections}, since whenever the rank of $P_1$ is zero we must have $P_1=0$, and hence $Tr(P_1 P_2) = Tr(P_1 P_3) = 0$, leaving us to determine the possible values of $Tr(P_2 P_3)$. We must further consider two cases, depending on whether or not $d \geq 3n + k + k' + j$. If this holds, then \[ \tfrac{d-j-n}{d} S_{d-j-n}(0, n+k, n+k') = (0,0,[0, \tfrac{n+k}{d}]) \] by Lemma \ref{traceProjections}.  In the case when $d < 3n + k + k' + j$, we have \[ \tfrac{d-j-n}{d} S_{d-j-n}(0,n+k,n+k') = (0,0, [\tfrac{3n + k + k' + j-d}{d}, \tfrac{n+k}{d}]) \] by Lemma \ref{traceProjections} again. Finally we show that each of the sets $(0,0,[0,\tfrac{n+k}{d}])$ and $(0,0,[\tfrac{3n + k + k' + j - d}{d}, \tfrac{n+k}{d}])$ are contained in the convex hull of $B_1$ and $B_2$. In the case when $d \geq 3n + k + k' + j$, $B_1$ is the singleton set containing $(0,0,0)$, so we need only show that $B_2$ contains the correlation $(0,0,\tfrac{n+k}{d})$. Indeed, considering the $2 \times 2$ matrices $P_1 = 1 \oplus 0$ and $P_2=P_3 = 0 \oplus 1$ we see that this holds. When $d < 3n + k + k' + j$, we have \[ B_1 = \tfrac{6n+2k+2k'-2d}{d}S_2(1). \] Again using $P_1 = 1 \oplus 0$ and $P_2=P_3 = 0 \oplus 1$, we see that \[ (0,0, \tfrac{3n+k+k'-d}{d}) \in B_1. \] Observing that $3n + k + k' - d < 3n+k+k'+j - d \leq n+k$ we see that \[ (0,0,[\tfrac{3n + k + k' + j}{d}, \tfrac{n+k}{d}]) \in \co\{B_1, B_2\}, \] concluding the case $j \leq m - n$.

We now consider the case $j > m-n$. In this case, $l_j = m - j$ and hence $n_j = n-m+j$ and $d_j = d-m$. We are left to consider the set \[ \tfrac{d_j}{d} S_{d_j}(n_j,n+k,n+k) = \tfrac{d-m}{d} S_{d-m}(n-m+j, n+k, n+k'). \] Since $(d-m) - (n+k) = n+k'$, we have \[ \tfrac{d-m}{d} S_{d-m}(n-m+j, n+k, n+k') = \tfrac{d-m}{d} \pi_3 S_{d-m}(n-m+j,n+k,n+k) \] where $\pi_3$ is the affine map given by reversing the outcome of the third experiment. Applying Lemma \ref{propTypeIISwap}, we see that $\tfrac{d-m}{d} S_{d-m}(n-m+j, n+k, n+k')$ is a subset of \begin{equation} \label{swapEqn} \co \{ \max(0, \tfrac{6n + 2j + 4k - 2d}{d}) S_2(1) + \vec{a}_j , \tfrac{2(n-m+j)}{d} S_2(1) + (0,0,[0, \tfrac{k+m-j}{d}]) \} \end{equation} where \begin{equation} \label{swapAjay} \vec{a}_j = (0, \min(\tfrac{n-m+j}{d}, \tfrac{k'-k}{d}), \min(\tfrac{n-m+j}{d}, \tfrac{k'-k}{d}) + \tfrac{k + m - j}{d}). \end{equation} We must consider two cases depending on whether or not $d \geq 3n + 2k + j$. When this does occur, we have $n - m + j \leq d - m - 2n - 2k = k'-k$. Thus \[ \max(0, \tfrac{6n + 2j + 4k - 2d}{d}) S_2(1) + \vec{a}_j = \{ \vec{a}_j \} =  \{ (0,\tfrac{n-m+j}{d}, \tfrac{n+k}{d}) \}. \] We will show that the correlation $\vec{a}_j$ is in the convex hull of $B_2$ and $B_3$. First observe that $B_2$ contains the correlation $(0,0,\tfrac{n+k}{d})$ (as shown in the previous paragraph). In the case when $n \leq k' - k$, $B_3$ is the single set containing the correlation $(0,\tfrac{n}{d}, \tfrac{n+k}{d})$. Since $0 \leq n-m+j \leq n$, we see that $\vec{a}_j \in \co \{B_2, B_3\}$. In the case when $n > k'-k$ we have \[ B_3 = \tfrac{2(n+k-k')}{d} S_2(1) + (0,\tfrac{k'-k}{d}, \tfrac{k'}{d}). \] Since $(0,0,1/2) \in S_2(1)$, we see that \[ (0,0,\tfrac{n+k-k'}{d}) + (0,\tfrac{k'-k}{d}, \tfrac{k'}{d}) = (0, \tfrac{k'-k}{d}, \tfrac{n+k}{d}) \in B_3. \] Since $0 \leq n-m+j \leq k-k'$ we see that $\vec{a}_j \in \co\{B_2,B_3\}$ again. We conclude the case $d \geq 3n+2k+j$ by showing that \[ \tfrac{2(n-m+j)}{d} S_2(1) + (0,0,[0, \tfrac{k+m-j}{d}]) \subseteq \co\{B_1, B_2\}. \] To see this, first observe that since $n - m = 3n + k + k' - d$, we have \[ \max(0,2(3n + k' + k' - d)) \leq 2(n - m + j) \leq 2n, \] implying that $\tfrac{2(n-m+j)}{d} S_2(1)$ lies in the convex hull of $B_1$ and $B_2$. Also, \[ \tfrac{2(n-m+j)}{d} S_2(1) + (0,0, \tfrac{k+m-j}{d}) \subseteq B_2 \] since any correlation in the left hand side can be implemented with projections of the form \[ P_i = (\hat{P}_i \otimes I_{n-(m-j)}) \oplus (\delta_{i,2}+\delta_{i,3})I_{m-j} \oplus \delta_{i,1} I_{m-j} \oplus (\delta_{i,2} + \delta_{i,3}) I_k \] where the $\hat{P}_i$ are $2 \times 2$ rank-one projections. Since the operators \[ (\hat{P}_i \otimes I_{n-(m-j)}) \oplus (\delta_{i,2}+\delta_{i,3})I_{m-j} \oplus \delta_{i,1} I_{m-j}\]  are $2n \times 2n$ projections of rank $n$, we see that the corresponding correlation belongs to $\tfrac{2n}{d}S_2(1) + (0,0,\tfrac{k}{d})$, concluding the case $d \geq 3n+2k+j$.

We conclude the proof by considering the case when $j > m - n$ and $d < 3n + 2k + j$. In this case we have $n \geq n - m + j > d - m -2n -2k = k'-k$, so that $\tfrac{d_j}{d} S_{d_j}(n_j,n+k,n+k')$ is a subset of \[ \co\{ \tfrac{6n+2j+4k-2d}{d}S_2(1) + \vec{a}_j, \tfrac{2(n-m+j)}{d} S_2(1) + (0,0,[0,\tfrac{k+m-j}{d}]) \} \] where $\vec{a}_j = (0,\tfrac{k'-k}{d}, \tfrac{k+m-j}{d})$, using equations (\ref{swapEqn}) and (\ref{swapAjay}), respectively, and $B_1$ and $B_3$ equal \[ \tfrac{6n + 2k + 2k' - 2d}{d} S_2(1) \quad \text{ and } \quad \tfrac{2(n-(k'-k))}{d} S_2(1) + (0, \tfrac{k'-k}{d}, \tfrac{k'}{d}), \] respectively. Repeating the arguments in the previous paragraph, we see that \[ \tfrac{2(n-m+j)}{d} S_2(1) + (0,0,[0,\tfrac{k+m-j}{d}]) \subset \co\{B_1,B_2\}. \] We finish the proof by showing that $\tfrac{6n+2j+4k-2d}{d}S_2(1) + \vec{a}_j \subseteq B_3$. Indeed, any correlation in $\tfrac{6n+2j+4k-2d}{d}S_2(1) + (0,\tfrac{k'-k}{d}, \tfrac{k+m-j}{d})$ can be implemented with projections of the form \[ P_i = (\hat{P}_i \otimes I_{n - (k'-k) + j - m}) \oplus (\delta_{i,2} + \delta_{i,3})I_{m-j} \oplus \delta_{i,1} I_{m-j} \oplus (\delta_{i,1} + \delta_{i,3}) I_{k'-k} \oplus (\delta_{i,2} + \delta_{i,3}) I_{k'} \] where $\hat{P}_i$ are $2 \times 2$ projections, since $6n+2j+4k-2d = 2(n-(k'-k)+j-m)$. Since  \[ (\hat{P}_i \otimes I_{n - (k'-k) + j - m}) \oplus (\delta_{i,2} + \delta_{i,3})I_{m-j} \oplus \delta_{i,1} I_{m-j} \] is a projection of rank $n-(k'-k)$ and size $2(n - (k'-k))$ by $2(n-(k'-k))$, we see that the correlation implemented by these projections is in $B_3$. \end{proof}

\subsection{Description of $C_q^s(3,2)$}

We are finally ready to describe $C_q^s(3,2)$. It suffices to describe an arbitrary slice $\hat{S}_{\vec{r}}[ C_q^s(3,2) ]$. Since every slice is an affine image of a standard slice, it suffices to provide a description for standard vectors $\vec{r}$ only. This is achieved in the following theorem.

\begin{theorem} \label{mainTheorem} Let $(r_1,r_2,r_3) = \vec{r} \in \mathbb{R}^3$ such that $0 \leq r_1 \leq r_2 \leq r_3 \leq 1/2$. Then the slice $\hat{S}_{\vec{r}}[ C_q^s(3,2)]$ is equal to the convex hull of the three sets \[ D_1^{\vec{r}} := 2\max(0,r_1 + r_2 + r_3 - 1)S_2(1), \quad D_2^{\vec{r}} := 2r_1 S_2(1) + (0,0,[0,r_2-r_1]), \] and \[ D_3^{\vec{r}} := 2\max(0,r_1+r_2-r_3) S_2(1) + (0, \min(r_1,r_3 - r_2), \min(r_2,r_3 - r_1)). \] Consequently, $C_q^s(3,2)$ is a topologically closed set. \end{theorem}

\begin{proof} Let $\vec{s} = (s_1, s_2, s_3) \in \mathbb{R}^3$ be chosen such that $s_1, s_2, s_3 \in [0,1]$. Then $\hat{S}_{\vec{s}}[C_q^s(3,2)]$ is non-empty. By swapping experiments and reversing outcomes, we may transform $\hat{S}_{\vec{s}}[C_q^s(3,2)]$ via some affine transformation $\pi$ into $\hat{S}_{\vec{r}}[C_q^s(3,2)] = \pi \hat{S}_{\vec{s}}[C_q^s(3,2)]$, where $\vec{r}$ is standard - i.e., $r_1 \leq r_2 \leq r_3 \leq 1/2$.

By Theorem \ref{densityThm}, the closure of the set of standard slices of $C_{max}^s(3,2)$ is equal to the set of standard slices of the closure of $C_q^s(3,2)$. By Propositions \ref{propTypeI}, \ref{propTypeII}, and \ref{propositionTypeIII}, $\hat{S}_{\vec{r'}}[ C_{max}^s] \subseteq \co \{D_1^{\vec{r'}}, D_2^{\vec{r'}}, D_3^{\vec{r'}} \}$ for each rational standard vector $\vec{r'}$. It follows that the $\vec{r}$ slice of the closure of $\cup_{\vec{r'}} \hat{S}_{\vec{r'}}[ C_{max}^s(3,2)]$ is precisely $\co \{D_1^{\vec{r}}, D_2^{\vec{r}}, D_3^{\vec{r}} \}$. Thus, $\hat{S}_{\vec{r}}[ C_{qa}^s(3,2) ] \subseteq \co \{D_1^{\vec{r}}, D_2^{\vec{r}}, D_3^{\vec{r}} \}$.

We conclude by showing that $\co \{D_1^{\vec{r}}, D_2^{\vec{r}}, D_3^{\vec{r}} \} \subseteq \hat{S}_{\vec{r}}[C_q^s(3,2)]$. To achieve this, we will show that $D_i^{\vec{r}} \subseteq \hat{S}_{\vec{r}}[C_q^s(3,2)]$ for each $i=1,2,3$. To do this, it suffices to demonstrate projections $P_i$ on a finite-dimensional $C^*$-algebra $\frak{A}$ with a state $\tau$ such that each correlation in $D_i^{\vec{r}}$ can be realized as $(\tau(P_1 P_2), \tau(P_1 P_3), \tau(P_2 P_3))$ with $\tau(P_i) = r_i$. To do this, we consider each $D_i^{\vec{r}}$ separately.

We begin with $D_1^{\vec{r}}$. In the case when $r_1 + r_2 + r_3 \leq 1$, then we can implement $D_1^{\vec{r}} = \{(0,0,0)\}$ with projections \[ P_i = \delta_{i,1} \oplus \delta_{i,2} \oplus \delta_{i,3} \oplus 0 \in \mathbb{C} \oplus \mathbb{C} \oplus \mathbb{C} \oplus \mathbb{C} \] using the trace $\tau(a \oplus b \oplus c \oplus d) = r_1 a + r_2 b + r_3 c + (1-r_1-r_2-r_3) d$. In the case when $r_1 + r_2 + r_3 > 1$, we may implement $D_1^{\vec{r}} = 2(r_1 + r_2 + r_3 - 1)S_2(1)$ with projections of the form \[ P_i = \hat{P}_i \oplus \delta_{i,1} \oplus \delta_{i,2} \oplus \delta_{i,3} \in \mathbb{M}_2 \oplus \mathbb{C} \oplus \mathbb{C} \oplus \mathbb{C} \] where $\hat{P}_i$ is a rank one $2 \times 2$ projection, using the trace \begin{eqnarray} \tau(A \oplus b \oplus c \oplus d) & = & (r_1 + r_2 + r_3 - 1) \Tr(A) + (1-r_2 - r_3)b  \nonumber \\ & & + (1 - r_1 - r_3) c + (1 - r_1 - r_2) d. \nonumber \end{eqnarray}

Next we consider $D_2^{\vec{r}} = 2r_1 S_2(1) + (0,0,[0,r_2-r_1])$. All correlations in this set can be implemented with projections of the form \[ P_i = \hat{P}_i \oplus (\delta_{i,2}+\delta_{i,3})\hat{Q}_i \oplus \delta_{i,3} \oplus 0 \in \mathbb{M}_2 \oplus \mathbb{M}_2 \oplus \mathbb{C} \oplus \mathbb{C} \] using the trace \begin{eqnarray} \tau(A \oplus B \oplus c \oplus d) & = & r_1 \Tr(A) + (r_2 - r_1)\Tr(B) \nonumber \\ & &  + (r_3 - r_2) c + (1 - r_3 - r_2) d. \nonumber \end{eqnarray}

Finally we consider \[ D_3^{\vec{r}} = \max(0,r_1+r_2-r_3) S_2(1) + (0, \min(r_1,r_3 - r_2), \min(r_2,r_3 - r_1)). \] In the case when $r_1 + r_2 \leq r_3$, we have $D_3^{\vec{r}} = \{(0,r_1,r_2)\}$. This correlation can be implemented with projections of the form \[ P_i = (\delta_{i,1} + \delta_{i,3}) \oplus (\delta_{i,2} + \delta_{i,3}) \oplus \delta_3 \oplus 0 \in \mathbb{C} \oplus \mathbb{C} \oplus \mathbb{C} \oplus \mathbb{C} \] using the trace $\tau(a \oplus b \oplus c \oplus d) = r_1 a + r_2 b + (r_3 - r_1 - r_2)c + (1-r_3)d$. In the case $r_1 + r_2 > r_3$, we have \[ D_3^{\vec{r}} = 2(r_1 + r_2 - r_3)S_2(1) + (0, r_3 - r_2, r_3 - r_1). \] This can be implemented using projections of the form \[ P_i = \hat{P}_i \oplus (\delta_{i,1} + \delta_{i,3}) \oplus (\delta_{i,2} + \delta_{i,3}) \oplus 0 \in \mathbb{M}_2 \oplus \mathbb{C} \oplus \mathbb{C} \oplus \mathbb{C} \] with the trace \begin{eqnarray} \tau(A \oplus b \oplus c \oplus d) & = & (r_1 + r_2 - r_3) \Tr(A) + (r_3 - r_2) b \nonumber \\ & & + (r_3 - r_1) c + (1 - r_1 - r_2) d. \nonumber \end{eqnarray}

Since $D_i^{\vec{r}} \subseteq \hat{S}_{\vec{r}}[C_q^s(3,2)]$ for each $i=1,2,3$, we conclude that \[ \co \{D_1^{\vec{r}}, D_2^{\vec{r}}, D_3^{\vec{r}}\} \subseteq \hat{S}_{\vec{r}}[ C_q^s(3,2) ] \] and the proof is complete. \end{proof}


Since each slice is a region in $\mathbb{R}^3$, we may easily visualize them. Some examples are recorded in Figures 2 and 3.

\begin{figure}[h!]
\centering
\begin{subfigure}
    \centering
    \includegraphics[width=.3\textwidth]{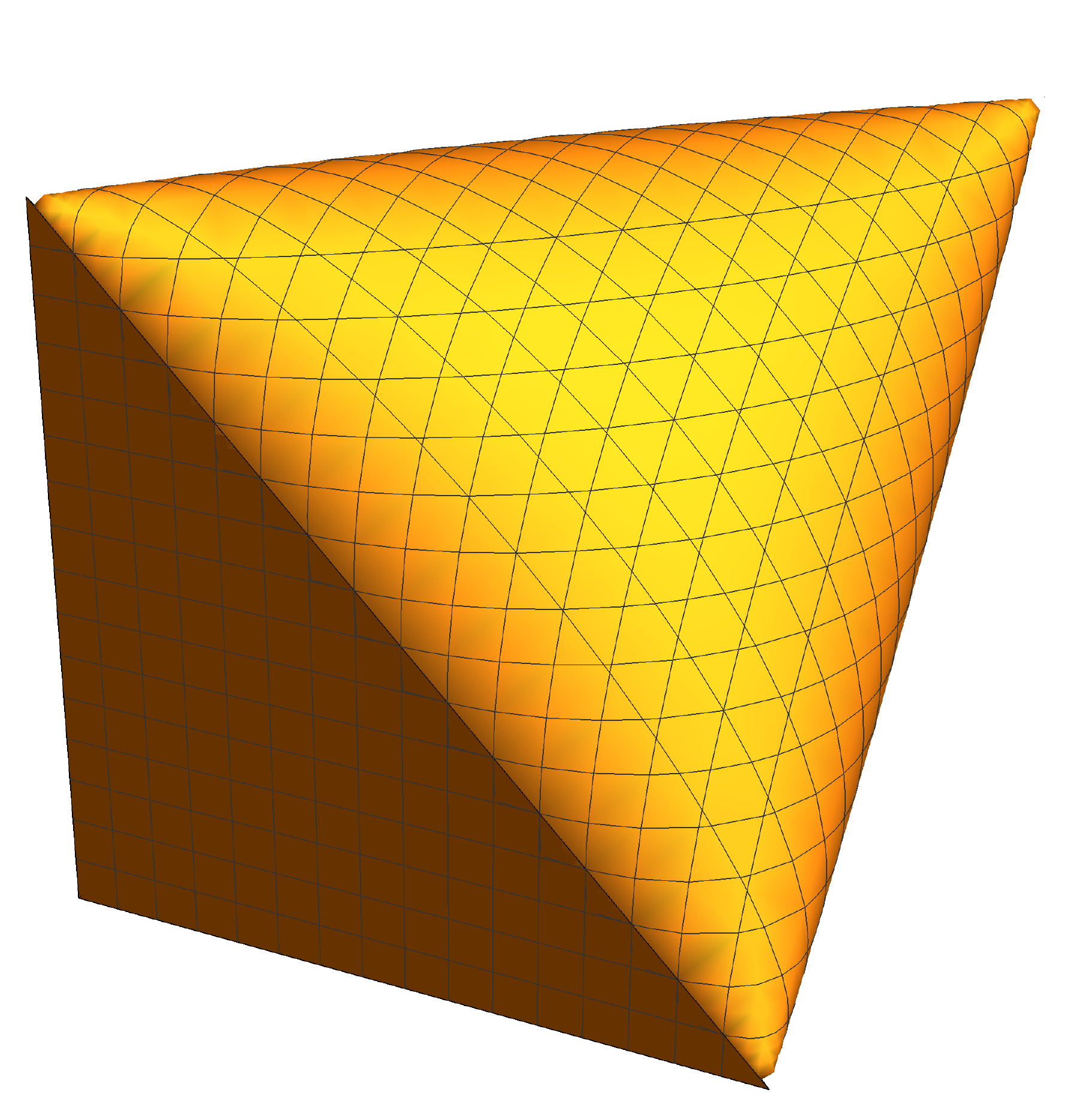}
\end{subfigure}
\begin{subfigure}
    \centering
    \includegraphics[width=.3\textwidth]{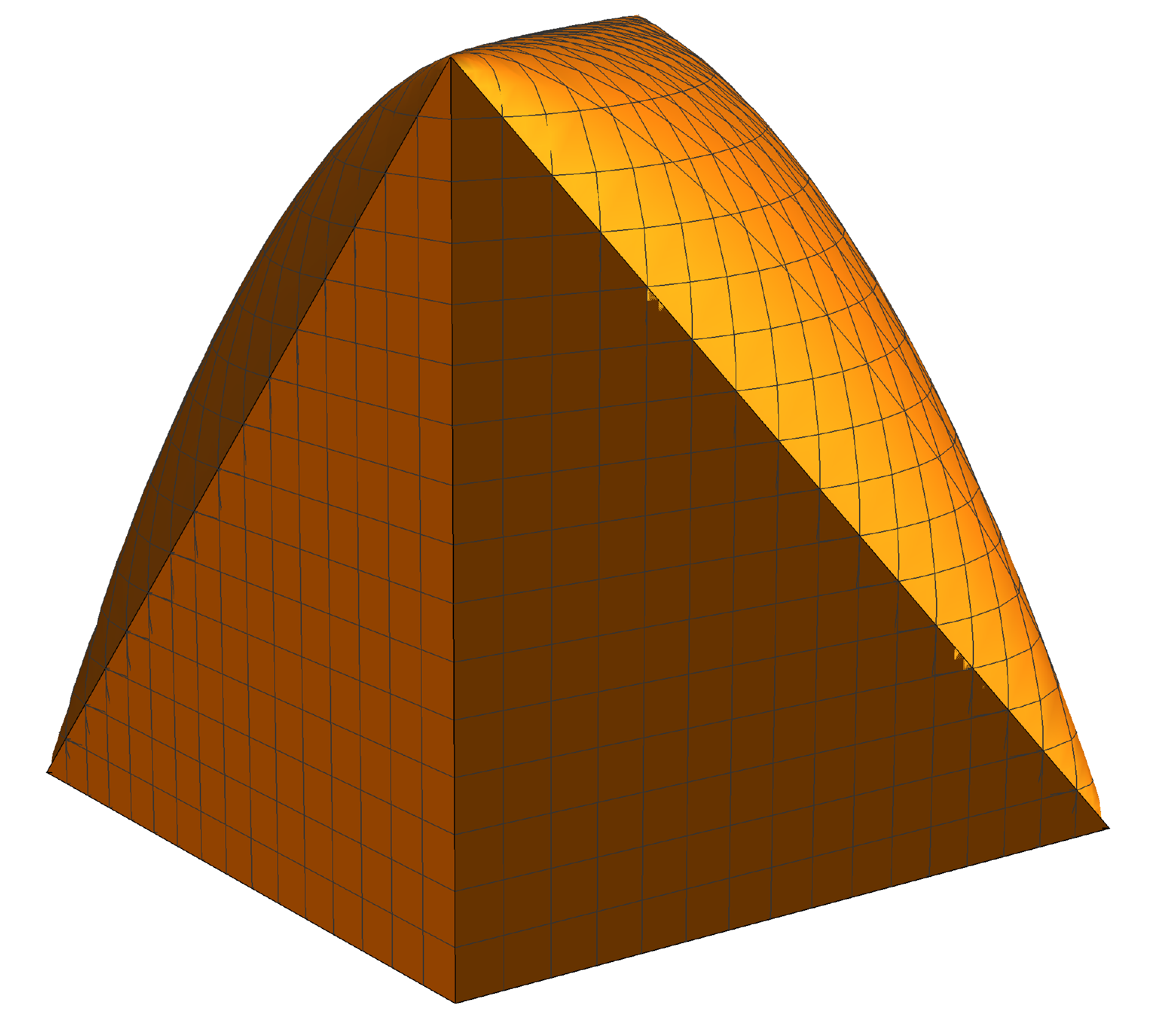}
\end{subfigure}
\caption{The $(1/3,1/3,1/3)$-slice of $C_q^s(3,2)$.}
\end{figure}

\begin{figure}[h!]
\centering
\begin{subfigure}
    \centering
    \includegraphics[width=.3\textwidth]{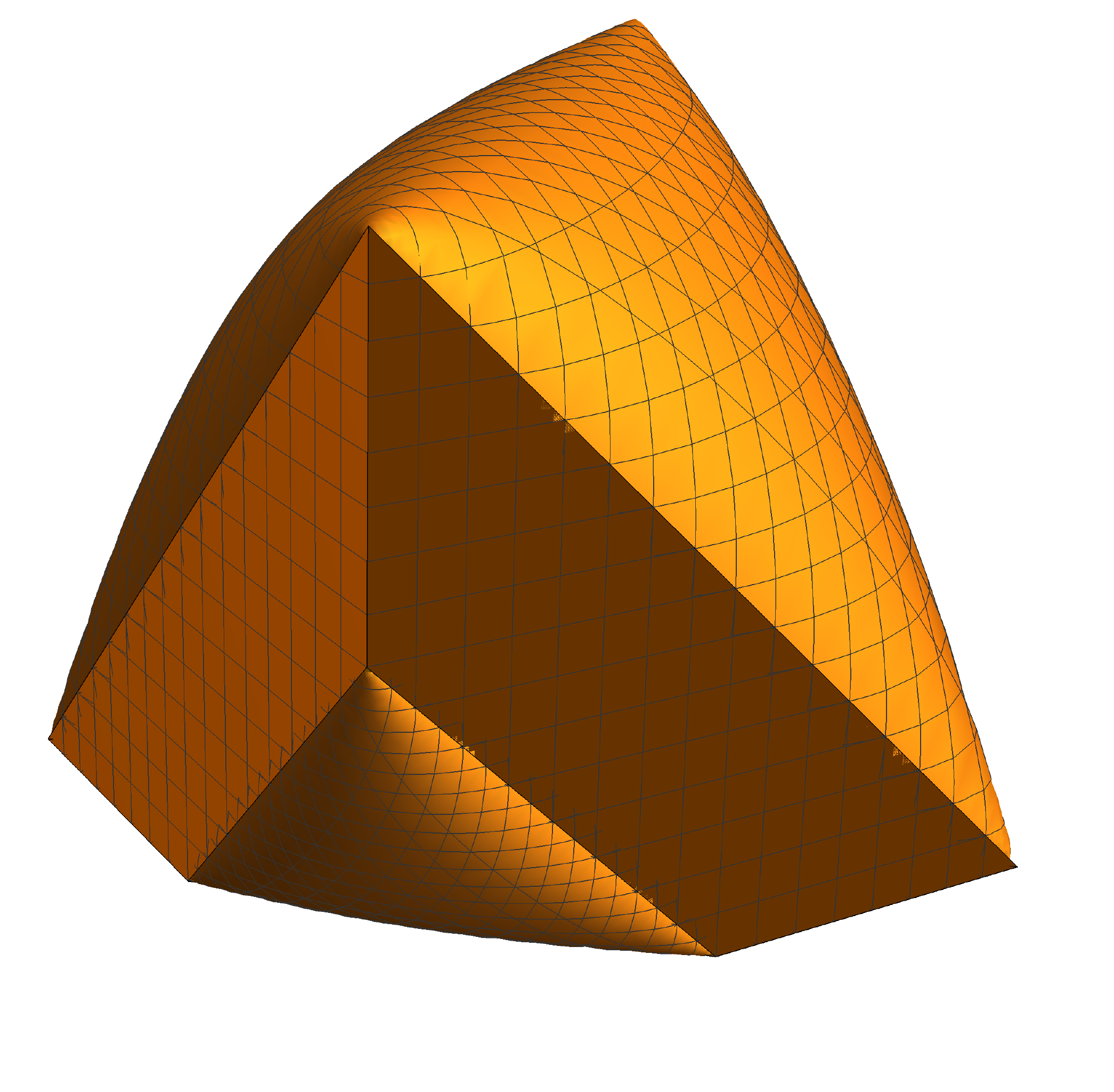}
\end{subfigure}
\begin{subfigure}
    \centering
    \includegraphics[width=.3\textwidth]{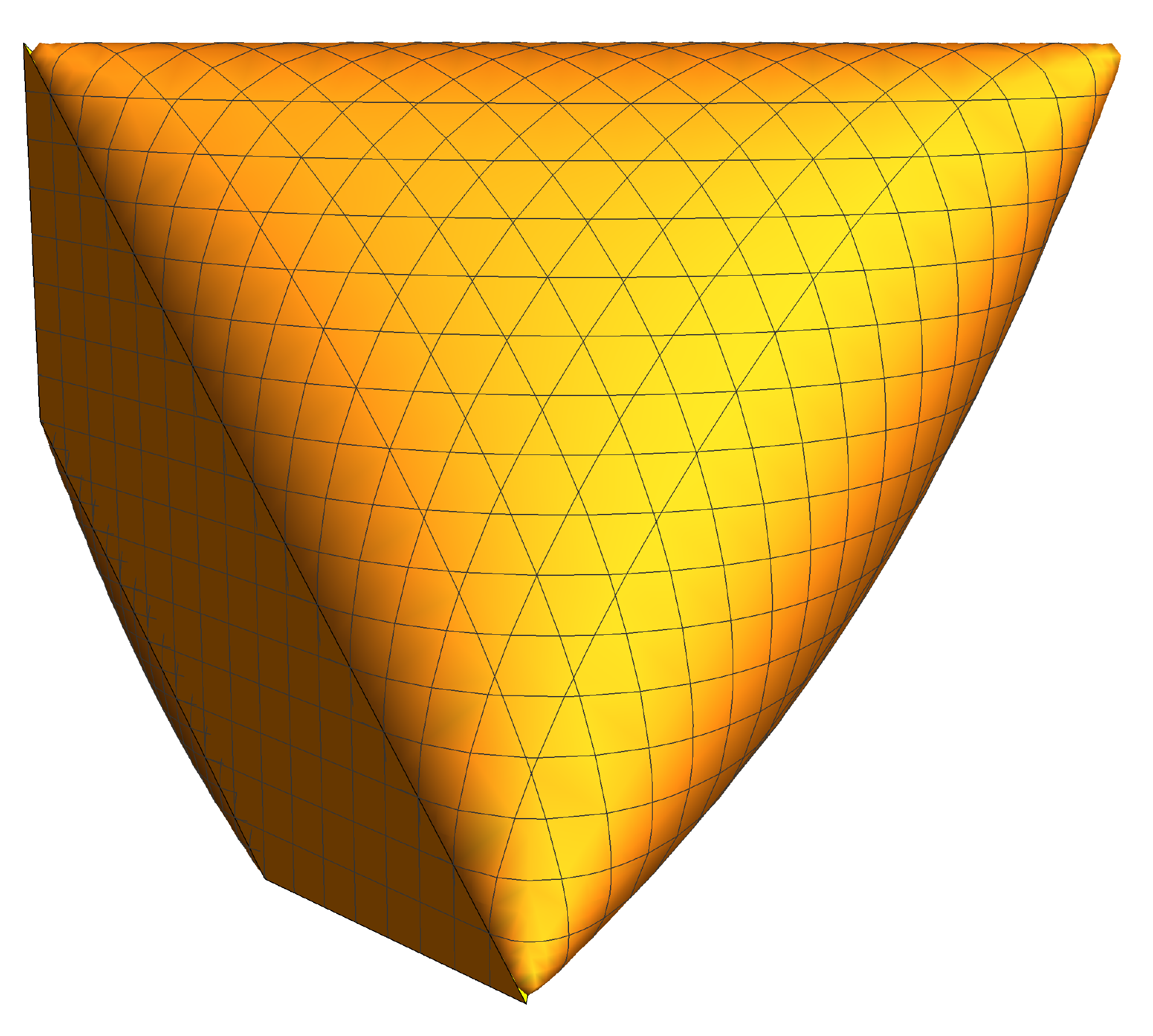}
\end{subfigure}
\caption{The $(2/5,2/5,2/5)$-slice of $C_q^s(3,2)$.}
\end{figure}

In the proof of Theorem \ref{mainTheorem}, we not only characterized the structure of $C_q^s(3,2)$, but we also saw how to build correlations in $C_q^s(3,2)$ using traces on finite-dimensional $C^*$-algebras and projection valued measures. This puts a bound on the dimension of Hilbert space required to implement correlations in $C_q^s(3,2)$ which we record in the following corollary.

\begin{corollary} If $p \in C_q^s(3,2)$, then there exists a $C^*$-subalgebra $\frak{A}$ of $\mathbb{M}_d$ with a trace $\tau$ such that $p(i,j|x,y) = \tau(E_{x,i} E_{y,j})$, where $\{E_{x,i}\}_{i=1,2}$ are projection valued measures in $\frak{A}$ and $d$ is no more than 16. \end{corollary}

\begin{proof} By examining the representations in the proof of Theorem \ref{mainTheorem}, we see that correlations in $D_i^{\vec{r}}$ can be produced using $C^*$-subalgebras $\frak{A}_i$ of $\mathbb{M}_{d_i}$ where $d_1 \leq 5$, $d_2 \leq 6$, and $d_3 \leq 5$. Convex combinations of these correlations can be built using traces on the direct sum $\frak{A}_1 \oplus \frak{A}_2 \oplus \frak{A}_3$, which is a subalgebra of $\mathbb{M}_d$ with $d \leq 16$. \end{proof}

\section{Concluding remarks}

In this final section, we discuss the question of whether or not synchronous quantum correlations coincide with the synchronous quantum commuting correlations in the three-experiment two-outcome setting. We first recall the definition of the quantum commuting correlations and a theorem describing how these correlations arise.

A correlation tensor $\{p(i,j|x,y)\}$ is called a \textbf{quantum commuting} correlation if there exists a Hilbert space $H$, projection valued measures $\{E_{x,i}\}_{i=1}^m$ and $\{F_{y,j}\}_{i=1}^m$ on $H$ satisfying $E_{x,i} F_{y,j} = F_{y,j} E_{x,i}$, and a state $\phi \in H$ such that \[ p(i,j|x,y) = \langle \phi, E_{x,i} F_{y,j} \phi \rangle. \] The set of all quantum commuting correlations is denoted by $C_{qc}(n,m)$. Quantum commuting correlations satisfying the synchronous condition $p(i,j|x,x)=0$ whenever $i \neq j$ also satisfy the following theorem, which generalizes Theorem \ref{PaulsenWinter}.

\begin{theorem}[Theorem 5.5 \cite{MR3460238}] \label{PaulsenWinterII} Let $p \in C_{qc}^s(n,m)$. Then there exists a $C^*$-algebra $\frak{A}$ and projection valued measures $\{E_{x,i}\}_{i=1}^m \subset \frak{A}$ and a tracial state $\tau$ on $\frak{A}$ such that \[ p(i,j|x,y) = \tau(E_{x,i} E_{y,j}). \] \end{theorem}

We define $C_{qa}(n,m)$ to be the closure of $C_q(n,m)$. By Theorem \ref{mainTheorem}, we know that $C_q^s(3,2)$ is closed. Moreover, it was shown by Kim, Paulsen and Schafhauser in \cite{MR3776034} that $C_{qa}^s(n,m)$ is equal to the closure of $C_q^s(n,m)$. This, together with Theorem \ref{mainTheorem}, implies the following obvious corollary.

\begin{corollary} The sets $C_q^s(3,2)$ and $C_{qa}^s(3,2)$ coincide. \end{corollary}

Another theorem of Kim-Paulsen-Schafhauser (Theorem 3.6 of \cite{MR3776034}) makes a fairly explicit connection between the set $C_{qa}^s(n,m)$ and Connes' conjecture. Roughly, it says that for any $p \in C_{qa}^s(n,m)$, Theorem \ref{PaulsenWinterII} holds with the $C^*$-algebra $\frak{A}$ being $\mathcal{R}^\omega$, a tracial ultrapower of the hyperfinite $II_1$ factor $\mathcal{R}$. It follows that Theorem \ref{mainTheorem} characterizes the set of correlations in $C_{qc}^s(3,2)$ which arise from $C^*$-algebras which embed in a trace-preserving way into $\mathcal{R}^\omega$.



It was shown by Ozawa (see Theorem 36 of \cite{MR3067294}) that the statement $C_{qa}(n,m) = C_{qc}(n,m)$ for all $n$ and $m$ is equivalent to Connes' embedding problem. Therefore, if Connes' embedding problem has an affirmative answer, we would have $C_{qc}^s(3,2) = C_q^s(3,2)$, by Theorem \ref{mainTheorem}. To verify that this holds, one would potentially need to compute $C_{qc}^s(3,2)$ explicitly. Indeed, it is known that $p \in C_{qc}^s(n,2)$ if and only if there exists a tracial state $\tau$ on the full group $C^*$-algebra $C^*(*_{i=1}^n \mathbb{Z}_2)$ such that $p(i,j|x,y) = \tau(E_{x,i} E_{y,j})$ where the projection valued measures $E_{x,i}$ are the canonical projections in $C^*(\mathbb{Z}_2) = \mathbb{C} \oplus \mathbb{C}$ on the $x$-th summand of the free product. When $n=2$, we are talking about $C^*(\mathbb{Z}_2 * \mathbb{Z}_2)$. This $C^*$-algeabra is very well understood (see, for example, Remark 3.6 of \cite{FritzKirchberg}) in part owing to the fact that the group $\mathbb{Z}_2 * \mathbb{Z}_2$ is amenable. In fact, this $C^*$-algebra is known to be isomorphic to a $C^*$-subalgebra of $C([0,1],\mathbb{M}_2)$ (continuous functions from $[0,1]$ to $\mathbb{M}_2$), from which one can easily conclude that $C_q(2,2) = C_{qc}(2,2)$ (for example, see exercise VI.6 of \cite{MR1402012}). In the case $n=3$, we are working with $C^*(\mathbb{Z}_2 * \mathbb{Z}_2 * \mathbb{Z}_2)$. This $C^*$-algebra is far less understood, since the group $\mathbb{Z}_2 * \mathbb{Z}_2 * \mathbb{Z}_2$ is not amenable. Therefore it is not so clear how one would decide whether or not $C_q^s(3,2) = C_{qc}^s(3,2)$ using Theorem \ref{PaulsenWinterII} with $\frak{A} = C^*(\mathbb{Z}_2 * \mathbb{Z}_2 * \mathbb{Z}_2)$.

Another correlation set we have not yet mentioned is the set of \textbf{vectorial} correlations, denoted $C_{vect}(n,m)$. Without providing an explicit description (for example, see Definition 2.6 of \cite{DPP1}), we note that \[ C_q(n,m) \subseteq C_{qc}(n,m) \subseteq C_{vect}(n,m). \]  Thus, another potential method of proving $C_q^s(3,2) = C_{qc}^s(3,2)$ would be to show that $C_q^s(3,2) = C_{vect}^s(3,2)$. Unfortunately, this is false.

\begin{corollary}[Corollary 3.2 of \cite{DeltaGame}] The sets $C_q^s(3,2)$ and $C_{vect}^s(3,2)$ do note coincide. \end{corollary}

\noindent Dykema-Paulsen-Prakash prove this theorem using the theory non-local games. They develope a synchronous generalization of the $I_{3322}$ game called the $\Delta$ game. They explicitly compute the values associated to the $\Delta$ game over an affine slice of $C_r^s(3,2)$ (though not one of the slices we have considered) for $r \in \{q, qc, vect\}$. They show that these values coincide on $C_q^s(3,2)$ and $C_{qc}^s(3,2)$ but differ from the values attained on $C_{vect}^s(3,2)$.


From these observations, it seems that the following remains open.

\begin{question} Does $C_q^s(3,2) = C_{qc}^s(3,2)$? \end{question}

\noindent An affirmative answer would perhaps be interesting in light of the discussion above concerning amenability. A negative answer would solve Connes' embedding problem.

\subsection*{Acknowledgements} We thank Elie Alhajjar for many helpful conversations and suggestions in the early stages of this work. In particular, we should acknowledge that an early version of Proposition \ref{PropEasyCase} was first proven jointly by Elie Alhajjar and the author, and that this early exploration inspired much of the work that followed. We also thank Vern Paulsen and William Slofstra for pointing out the connection to the $I_{3322}$ game and the reference \cite{DeltaGame}.

\bibliographystyle{plain}
\bibliography{references_v2}

\begin{thebibliography}{10}

\bibitem{MaxEntangle}
E.~Alhajjar and T.~Russell.
\newblock Maximally entangled correlation sets.
\newblock {\em To appear, Houston J. Math.}, 2019.

\bibitem{BellTheoremExperiment}
Alain Aspect, Philippe Grangier, and G\'erard Roger.
\newblock Experimental tests of realistic local theories via bell's theorem.
\newblock {\em Phys. Rev. Lett.}, 47:460--463, Aug 1981.

\bibitem{bell_epr}
J.~S. Bell.
\newblock On the einstein podolsky rosen paradox.
\newblock {\em Physics Physique Fizika}, 1:195--200, Nov 1964.

\bibitem{bb84}
C.~Bennett and G.~Brassard.
\newblock Quantum cryptography: Public key distribution and coin tossing.
\newblock volume 560, pages 175--179, 01 1984.

\bibitem{CHSH}
John~F. Clauser, Michael~A. Horne, Abner Shimony, and Richard~A. Holt.
\newblock Proposed experiment to test local hidden-variable theories.
\newblock {\em Phys. Rev. Lett.}, 23:880--884, Oct 1969.

\bibitem{Collins_2004}
Daniel Collins and Nicolas Gisin.
\newblock A relevant two qubit bell inequality inequivalent to the {CHSH}
  inequality.
\newblock {\em Journal of Physics A: Mathematical and General},
  37(5):1775--1787, jan 2004.

\bibitem{MR1402012}
K.~R. Davidson.
\newblock {\em {$C^*$}-algebras by example}, volume~6 of {\em Fields Institute
  Monographs}.
\newblock American Mathematical Society, Providence, RI, 1996.

\bibitem{DPP1}
K.~Dykema, V.~I. Paulsen, and J.~Prakash.
\newblock Non-closure of the set of quantum correlations via graphs.
\newblock {\em Comm. Math. Phys.}, 365(3):1125--1142, 2019.

\bibitem{DeltaGame}
K.~Dykema, V.~I. Paulsen, and Jitendra Prakash.
\newblock The delta game.
\newblock {\em Quantum Information {\&} Computation}, 18(7{\&}8):599--616,
  2018.

\bibitem{MR3432742}
K.~J. Dykema and V.~I. Paulsen.
\newblock Synchronous correlation matrices and {C}onnes' embedding conjecture.
\newblock {\em J. Math. Phys.}, 57(1):015214, 12, 2016.

\bibitem{EPR}
A.~Einstein, B.~Podolsky, and N.~Rosen.
\newblock Can quantum-mechanical description of physical reality be considered
  complete?
\newblock {\em Phys. Rev.}, 47:777--780, May 1935.

\bibitem{FritzKirchberg}
T.~Fritz.
\newblock Tsirelson's problem and kirchberg's conjecture.
\newblock {\em Reviews in Mathematical Physics}, 24(05):1250012, 2012.

\bibitem{Froissart1981}
M.~Froissart.
\newblock Constructive generalization of bell's inequalities.
\newblock {\em Il Nuovo Cimento B (1971-1996)}, 64(2):241--251, Aug 1981.

\bibitem{PhysRevA.97.022104}
K.~Goh, J.~Kaniewski, E.~Wolfe, T.~V\'ertesi, X.~Wu, Y.~Cai, Y.~Liang, and
  V.~Scarani.
\newblock Geometry of the set of quantum correlations.
\newblock {\em Phys. Rev. A}, 97:022104, Feb 2018.

\bibitem{JMPPSW2011}
M.~Junge, M.~Navascues, C.~Palazuelos, D.~Perez-Garcia, V.~Scholz, and
  R.~Werner.
\newblock Connes' embedding problem and tsirelson's problem.
\newblock {\em J. Math. Phys.}, 52(1):012102, 2011.

\bibitem{MR3776034}
S.~Kim, V.~I. Paulsen, and C.~Schafhauser.
\newblock A synchronous game for binary constraint systems.
\newblock {\em J. Math. Phys.}, 59(3):032201, 17, 2018.

\bibitem{Lackey}
B.~Lackey and N.~Rodrigues.
\newblock Nonlocal games, synchronous correlations, and bell inequalities.
\newblock {\em arxiv}, abs/1707.06200, 2017.

\bibitem{PerfectStrategies}
M.~Lupini, L.~Mancinska, V.I. Paulsen, D.E. Roberson, G.~Scarpa, S.~Severini,
  I.G. Todorov, and A.~Winter.
\newblock Perfect strategies for non-signalling games.
\newblock {\em arxiv}, abs/1804.06151, 2018.

\bibitem{MR3067294}
N.~Ozawa.
\newblock About the {C}onnes embedding conjecture: algebraic approaches.
\newblock {\em Jpn. J. Math.}, 8(1):147--183, 2013.

\bibitem{osti_21448443}
Karoly~F. Pal and Tamas Vertesi.
\newblock Maximal violation of a bipartite three-setting, two-outcome bell
  inequality using infinite-dimensional quantum systems.
\newblock {\em Physical Review. A}, 82(2), 8 2010.

\bibitem{MR3460238}
V.~I. Paulsen, S.~Severini, D.~Stahlke, I.~G. Todorov, and A.~Winter.
\newblock Estimating quantum chromatic numbers.
\newblock {\em J. Funct. Anal.}, 270(6):2188--2222, 2016.

\bibitem{Slofstra1}
W.~Slofstra.
\newblock The set of quantum correlations is not closed.
\newblock {\em Forum Math. Pi}, 7:e1, 41, 2019.

\bibitem{ThinhStructure}
L.~Thinh, A.~Varvitsiotis, and Y.~Cai.
\newblock Structure of the set of quantum correlators via semidefinite
  programming.
\newblock {\em arxiv}, abs/1809.10886, 2018.

\bibitem{Tsirel'son1987}
B.~S. Tsirel'son.
\newblock Quantum generalizations of bell's inequality.
\newblock {\em Letters in Mathematical Physics}, 4(2):93--100, Mar 1980.

\bibitem{VidickI3322}
T.~Vidick and S.~Wehner.
\newblock More nonlocality with less entanglement.
\newblock {\em Phys. Rev. A}, 83:052310, May 2011.

\end{thebibliography}

\end{document}